\documentclass[a4paper,UKenglish,cleveref,autoref, thm-restate]{lipics-v2021}

\usepackage{tikz}
\usetikzlibrary{arrows.meta}
\newcommand{\dist}{\operatorname{dist}}
\newtheorem{cma}{\bf Reduction Rule CMA}
\newtheorem{defn}{Definition}%
\title{Parameterized Complexity of Connected Network Microaggregation: The Role of Cluster Size}

\titlerunning{Connected Network Microaggregation} 

\author{Ajinkya Gaikwad}{Czech Technical University in Prague
Prague, Czech Republic}{ajinkya.gaikwad@fit.cvut.cz}{https://orcid.org/0000-0002-7514-0708}{}

\author{Dušan Knop}{Czech Technical University in Prague
Prague, Czech Republic}{Dusan.Knop@fit.cvut.cz}{https://orcid.org/0000-0003-2588-5709}{}

\author{Tomáš Valla}{Czech Technical University in Prague
Prague, Czech Republic}{Tomas.Valla@fit.cvut.cz}{https://orcid.org/0000-0003-1228-7160}{}

\authorrunning{A. Gaikwad}

\Copyright{Ajinkya Gaikwad, Dušan Knop, and Tomáš Valla}



\category{} 

\relatedversion{} 

\acknowledgements{Research supported by the Czech Science Foundation Grant no.~24-12046S. This work was co-funded by the European Union under the project Robotics and Advanced Industrial Production (reg. no.~CZ.02.01.01/00/22\_008/0004590).}

\nolinenumbers 

\EventEditors{John Q. Open and Joan R. Access}
\EventNoEds{2}
\EventLongTitle{42nd Conference on Very Important Topics (CVIT 2016)}
\EventShortTitle{CVIT 2016}
\EventAcronym{CVIT}
\EventYear{2016}
\EventDate{December 24--27, 2016}
\EventLocation{Little Whinging, United Kingdom}
\EventLogo{}
\SeriesVolume{42}
\ArticleNo{23}

\begin{document}

\maketitle


\begin{abstract}
Network microaggregation is a fundamental technique in statistical disclosure control, where vertices of a graph are partitioned into clusters satisfying size constraints and admitting a center within bounded distance. We study the parameterized complexity of the \emph{unweighted Connected Network Microaggregation} problem, focusing on the interplay between structural parameters and natural clustering parameters such as the distance bound $d$ and the cluster size gap $u-\ell$.

While the weighted variant is $\mathrm{W[1]}$-hard when parameterized by vertex cover, we show that the unweighted connected variant becomes fixed-parameter tractable for neighborhood diversity, and hence also for vertex cover. However, this tractability does not extend to more general structural parameters: the problem remains $\mathrm{W[1]}$-hard even for parameters such as vertex deletion to paths, stars, or cliques. We significantly strengthen these hardness results by showing that they persist even in highly permissive regimes, namely for all $d \ge 2$ and any fixed gap $u-\ell$. This demonstrates that the parameters $d$ and $u-\ell$ are ineffective in combination with structural parameters.

We show that augmenting structural parameters with the cluster size bound $u$ restores tractability, yielding fixed-parameter algorithms for parameters such as treewidth and cluster vertex deletion combined with $u$. Moreover, we establish that $u$ is essential for tractability, as the problem remains $\mathrm{W[1]}$-hard when structural parameters are considered alone.

From a kernelization perspective, we prove that the problem does not admit a polynomial kernel when parameterized by vertex cover unless $\mathrm{coNP} \subseteq \mathrm{NP/poly}$, and that this lower bound persists even when the distance constraint becomes vacuous. On the positive side, adding $u$ yields a polynomial kernel for vertex cover, while kernelization remains unlikely for more general structural parameters even when combined with $u$. Finally, we show that the problem remains NP-hard on graphs of bounded clique-width.
\end{abstract} 

\newpage
\setcounter{page}{1}

\section{Introduction}

In an era where large-scale relational data is routinely collected and analyzed, preserving privacy while maintaining data utility has become a fundamental challenge. A central tool in statistical disclosure control is \emph{microaggregation}~\cite{Domingo-Ferrer2009,Yan2022}, which partitions data into clusters of bounded size and releases aggregate representatives (centers) instead of individual data points. Importantly, centers are not required to belong to their corresponding clusters, and a single vertex may serve as a center for multiple clusters. This approach provides a principled way to achieve anonymity guarantees such as $k$-anonymity~\cite{10.1142/S0218488502001648}, while retaining structural information about the underlying data.

When the input is a network, this leads to the \textsc{Network Microaggregation} problem, where the vertices of a graph must be partitioned into clusters that satisfy size constraints and admit a center within bounded distance. A natural and practically motivated variant is \emph{Connected Network Microaggregation} (\textsc{CNMA}), where each cluster is additionally required to induce a connected subgraph. As observed in prior work~\cite{Blažej_Ganian_Knop_Pokorný_Schierreich_Simonov_2023}, the connectivity requirement fundamentally changes the behavior of the problem: unlike the classical setting, clusters can no longer be arbitrarily split, and upper bounds on cluster sizes become essential for controlling the quality of the anonymized output.

Despite its importance, the parameterized complexity of network microaggregation has only recently been systematically studied. The seminal work of Blažej et al.~\cite{Blažej_Ganian_Knop_Pokorný_Schierreich_Simonov_2023} established a near-complete complexity landscape for both the unrestricted and connected variants under a wide range of structural and problem-specific parameters. In particular, they showed that while \textsc{CNMA} is fixed-parameter tractable under several combined parameterizations (e.g., vertex cover plus cluster size bound), it remains $\mathrm{W[1]}$-hard when parameterized by vertex cover alone. This highlights an inherent computational barrier in handling the connectivity constraint.

\medskip
\noindent
\textbf{Our focus.}
In this work, we revisit the problem from a more refined perspective by restricting attention to the \emph{unweighted} setting and focusing exclusively on the \emph{connected} variant. This restriction is well-motivated: even in the unweighted case, \textsc{CNMA} remains $\mathrm{W[1]}$-hard when parameterized by several restricted structural parameters, and hence already captures the core algorithmic challenges of the problem.

Our primary goal is to investigate whether combining such structural parameters with natural input parameters governing the clustering—such as the distance bound $d$, the cluster size bound $u$, and the slack $u-\ell$—can lead to tractability. In particular, we systematically explore parameterizations that augment structural measures (e.g., vertex cover, treewidth, or deletion distance to simple graph classes) with these additional parameters, and examine whether such combinations yield fixed-parameter tractability or efficient preprocessing.

By isolating this variant and studying these combined parameterizations, we obtain a sharper and more nuanced understanding of the parameterized complexity landscape, identifying precisely when positive algorithmic results can be achieved and when inherent hardness persists.

\medskip
\noindent
\textbf{Our Contribution.}
As an initial observation, it was noted in previous work that \textsc{Connected Network Microaggregation} remains NP-hard even for fixed values of $d$ and $u$. This implies that bounding these parameters alone is insufficient for tractability, and one must additionally exploit structural properties of the underlying graph.

It is known from previous work that the weighted version of the problem is $\mathrm{W[1]}$-hard when parameterized by the vertex cover number. This naturally raises the question of whether the hardness persists in the unweighted setting. We answer this in the negative by showing that the unweighted connected variant is fixed-parameter tractable when parameterized by neighborhood diversity, and consequently also when parameterized by vertex cover.

Having established tractability for vertex cover, the next natural question is whether this can be extended to more general structural parameters such as treewidth. We show that this is not the case: even under strong restrictions such as vertex deletion to paths, stars, or cliques, the problem remains $\mathrm{W[1]}$-hard. Moreover, we strengthen these hardness results by showing that they continue to hold even when the distance parameter and the gap parameter are highly permissive, namely for all $d \ge 2$ and any fixed value of $u-\ell$. This shows that the parameters $d$ and $u-\ell$ are essentially useless when considered together with structural parameters: relaxing them does not lead to tractability.

This leads to the question of whether augmenting structural parameters with the upper bound $u$ can overcome this barrier. We answer this positively by designing a fixed-parameter algorithm when parameterized by treewidth and $u$. In addition, we show that while the problem is $\mathrm{W[1]}$-hard when parameterized by the cluster vertex deletion number alone, it becomes fixed-parameter tractable when parameterized by $\mathrm{cvd}+u$. Thus, $u$ plays a central role in recovering tractability once structural restrictions are present. We also show that the problem becomes FPT when parameterized by cluster vertex deletion set and the size of maximum clique when $d\geq 3$.

We then turn to the question of kernelization. While the problem is fixed-parameter tractable when parameterized by vertex cover or by vertex deletion to a clique, we show that it does not admit a polynomial kernel under either parameterization. We next ask whether adding $u$ enable us to get polynomial kernel. We answer this positively for vertex cover by showing that the problem admits a polynomial kernel when parameterized by $\mathrm{vc}+u$. However, this phenomenon does not extend to more general parameters: even when combined with $u$, the problem does not admit a polynomial kernel for parameterizations such as vertex deletion to stars or paths, unless $\mathrm{coNP} \subseteq \mathrm{NP/poly}$. These lower bounds hold even when the distance constraint becomes vacuous, that is, for all $d \ge u$.
Finally, we complement our parameterized results with classical hardness by showing that the problem remains NP-hard even on graphs of bounded clique-width (in particular, at most $4$), even when $d \ge 2$ and $u-\ell$ is bounded by a constant.

Taken together, our results reveal a clear and coherent picture: the parameters $d$ and $u-\ell$ do not contribute to tractability when combined with structural parameters, since hardness persists even in highly relaxed regimes. In contrast, augmenting structural parameters with the upper bound $u$ yields fixed-parameter tractability, and in the case of vertex cover, even polynomial kernels. This precisely delineates the boundary between tractability, hardness, and compressibility for \textsc{Connected Network Microaggregation}. Refer to Figure~\ref{overview} for an overview.

\medskip
\noindent
\textbf{Related Work on Network Microaggregation.}
\textsc{Network Microaggregation} problem is NP-hard in general and has been studied in weighted settings. Recent work initiated a systematic parameterized study of \textsc{NMA} and its connected variant,\textsc{CNMA}, under both structural parameters (such as treewidth and vertex cover) and input parameters (such as the distance bound $d$ and cluster size bound $u$).

From a parameterized perspective, recent work of Blažej et al.~\cite{Blažej_Ganian_Knop_Pokorný_Schierreich_Simonov_2023} has initiated the study of \textsc{NMA} and its connected variant, \textsc{CNMA}, under structural parameters of the input graph. The resulting complexity landscape is rich and nuanced. For \textsc{NMA}, the problem is NP-hard even for constant values of $d$ and $u$, and remains $\mathrm{W[1]}$-hard under several parameterizations involving $d$ or $d+u$, even when combined with structural parameters such as treewidth. On the positive side, fixed-parameter tractability can be achieved when parameterizing by vertex cover together with $d$ or $u$, and in particular when combining $d+u$ with structural parameters. In contrast, \textsc{CNMA} exhibits a more delicate behavior: while it remains $\mathrm{W[1]}$-hard under treewidth-based parameterizations, it becomes fixed-parameter tractable when parameterized by vertex cover together with $d$ or $u$, and in particular admits FPT algorithms when parameterized by $d+u$ combined with structural parameters.

These results highlight that connectivity significantly increases the complexity of the problem under certain parameterizations, while also enabling tractability in others. Overall, they demonstrate that neither structural parameters nor input parameters alone suffice to fully capture the complexity of the problem, and that their interaction plays a crucial role in determining tractability.

\medskip
\noindent
\textbf{Related Graph Partitioning and Clustering Problems.}
A large body of work studies graph partitioning~\cite{GDOWNEY2003209,10.1007/978-3-030-67731-2_23,10.1007/978-3-030-64843-5_6,10.1007/978-3-642-11269-0_10,Demaine2003FixedParameterAF} and clustering problems~\cite{10.1007/978-3-319-08783-2_24,Madathil2024ParameterizedAF,gaikwad2025parameterizedcomplexitysclubcluster,pmlr-v162-ganian22a,Ganian_Kanj_Ordyniak_Szeider_2020} that aim to divide the vertex set under structural or optimization constraints. 
Classical formulations, such as balanced partitioning, seek to partition vertices into $k$ parts while minimizing cut edges, often under size constraints. 

A closely related problem to \textsc{CNMA} is \textsc{Equitable Connected Partition} (\textsc{ECP})~\cite{10.1007/978-3-642-11269-0_10}, where the goal is to partition the vertices into $p$ connected parts whose sizes differ by at most one. 
In contrast, \textsc{CNMA} relaxes this near-equal size requirement by allowing cluster sizes in $[\ell,u]$ and additionally imposes a distance constraint by requiring each cluster to admit a center within distance $d$ of all its vertices. 
From a parameterized perspective, Blažej et al.~\cite{blazej_et_al:LIPIcs.MFCS.2024.29} showed that \textsc{ECP} is NP-hard in general and $\mathrm{W[1]}$-hard under several parameterizations, including pathwidth, feedback vertex set, and the number of parts.
On the positive side, it is fixed-parameter tractable for parameters such as vertex integrity and maximum leaf number, and other structural parameters such as modular width and distance to clique. 
These results are complemented by matching hardness results and algorithmic upper bounds, including XP and FPT algorithms based on structural decompositions and integer programming techniques.

Another closely related problem is \textsc{$k$-Center}, where one selects $k$ centers to minimize the maximum distance of any vertex to its nearest center. The problem is NP-hard~\cite{Vazirani2001} and admits a tight $2$-approximation~\cite{10.1145/5925.5933}. Parameterized hardness results show that it is $\mathrm{W[2]}$-hard~\cite{Demaine2003FixedParameterAF} with respect to $k$ and remains hard under several structural parameters, while FPT approximation schemes are known for parameters such as treewidth and clique-width~\cite{katsikarelis_et_al:LIPIcs.ISAAC.2017.50}.

The \textsc{Network Microaggregation} problem is closely related to the well-studied \textsc{r-Gather} clustering problem, introduced in the context of $k$-anonymity~\cite{10.1145/1798596.1798602}. In \textsc{r-Gather}, the goal is to partition the input into clusters of size at least $r$ while minimizing a distance-based objective. The problem is NP-hard, but admits a $2$-approximation in general metrics~\cite{10.1145/1798596.1798602}.
In the Euclidean setting, approximation guarantees depending on $r$ are known, including an $O(r^3)$-approximation~\cite{DomingoFerrer2008APA} and a $2$-approximation for $r=2$~\cite{10.1007/11930242_12}. Notably, $r$ is typically assumed to be small in applications, as it corresponds to the anonymity requirement.

The key difference to \textsc{Network Microaggregation} is the absence of an upper bound on the cluster size. While in the non-connected setting such an upper bound can often be implicitly enforced, this is no longer the case for the connected variant, where clusters cannot be arbitrarily split without violating connectivity. Consequently, the upper bound $u$ plays a crucial role in \textsc{Connected Network Microaggregation}, and forms a central parameter in our study.

In contrast, \textsc{Connected Network Microaggregation} combines multiple constraints simultaneously: clusters must satisfy size bounds, induce connected subgraphs, and admit a center within bounded distance. 
This combination places our problem at the intersection of graph partitioning, clustering, and domination-type problems, while introducing new algorithmic challenges that are not present in the above formulations.

\begin{figure*}[t]
\centering
\scalebox{0.72}{
\begin{tikzpicture}[>=stealth, every node/.style={font=\small}]

\begin{scope}[xshift=0cm]
\node[font=\bfseries] at (0,5.0) {without $u$};

\node[rectangle, draw, green!60!black, thick, align=center] 
(vc1) at (0,-2) {{\color{green!60!black}vc}$^{\color{black}\times}$};
\node[rectangle, draw, green!60!black, thick, align=center] 
(nd1) at (0.5,1) {$\textcolor{green!60!black}{\mathrm{nd}}^{\!\textcolor{black}{?}}$};
\node[rectangle, draw, black, align=center] (tc1) at (2.25,-0.5) {tc};
\node[rectangle, draw, black, align=center] (mw1) at (0.5,2.5) {mw};
\node[rectangle, draw, red, align=center] (cw1) at (0,4) {cw};
\node[rectangle, draw, orange!85!black, align=center] (pw1) at (-0.9,2) {pw};
\node[rectangle, draw, orange!85!black, align=center] (fvs1) at (-3.25,1) {fvs};
\node[rectangle, draw, orange!85!black, align=center] (tw1) at (-2.5,2.5) {tw};
\node[rectangle, draw, orange!85!black, align=center] (cvd1) at (2.25,2.5) {cvd};
\node[rectangle, draw, orange!85!black, align=center] (td1) at (-1.75,1) {td};

\node[rectangle, draw, orange!85!black, align=center] (vdp1) at (-4,-0.5) {vdp};
\node[rectangle, draw, orange!85!black, align=center] (fes1) at (-5.5,-0.5) {fes};
\node[rectangle, draw, orange!85!black, align=center] (vds1) at (-2.5,-0.5) {vds};
\node[rectangle, draw, black, align=center] (vi1) at (-1,-0.5) {vi};
\node[rectangle, draw, green!60!black, thick, align=center] 
(vdc1) at (1,-0.5) {$\textcolor{green!60!black}{\mathrm{vdc}}^{\textcolor{black}{\times}}$};

\draw[->] (fes1) -- (fvs1);
\draw[->] (vi1) -- (td1);
\draw[->] (vc1) -- (vi1);
\draw[->] (vc1) -- (tc1);
\draw[->] (nd1) -- (mw1);
\draw[->] (tc1) -- (mw1);
\draw[->] (mw1) -- (cw1);
\draw[->] (tc1) -- (cvd1);
\draw[->] (fvs1) -- (tw1);
\draw[->] (pw1) -- (tw1);
\draw[->] (tw1) -- (cw1);
\draw[->] (cvd1) -- (cw1);
\draw[->] (td1) -- (pw1);

\draw[->] (vc1) -- (nd1);
\draw[->] (vc1) -- (vds1);
\draw[->] (vc1) -- (vdp1);
\draw[->] (vdc1) -- (nd1);
\draw[->] (vdc1) -- (cvd1);
\draw[->] (vdp1) -- (fvs1);
\draw[->] (vds1) -- (td1);
\draw[->] (vds1) -- (fvs1);
\end{scope}


\begin{scope}[xshift=10.0cm]
\node[font=\bfseries] at (0,5.0) {with $u$};

\node[rectangle, draw, green!60!black, thick, align=center] 
(vc2) at (0,-2) {{\color{green!60!black}vc}$^{\color{black}\star}$};

\node[rectangle, draw, green!60!black, thick, align=center] 
(nd2) at (0.5,1) {$\textcolor{green!60!black}{\mathrm{nd}}^{\!\textcolor{black}{?}}$};

\node[rectangle, draw, green!60!black, thick, align=center] 
(tc2) at (2.25,-0.5) {$\textcolor{green!60!black}{\mathrm{tc}}^{\!\textcolor{black}{?}}$};

\node[rectangle, draw, black, align=center] (mw2) at (0.5,2.5) {mw};
\node[rectangle, draw, black, align=center] (cw2) at (0,4) {cw};
\node[rectangle, draw, green!60!black, thick, align=center] (pw2) at (-0.9,2) {pw};
\node[rectangle, draw, green!60!black, thick, align=center] (fvs2) at (-3.25,1) {fvs};
\node[rectangle, draw, green!60!black, thick, align=center] (tw2) at (-2.5,2.5) {tw};

\node[rectangle, draw, green!60!black, thick, align=center] 
(cvd2) at (2.25,2.5) {$\textcolor{green!60!black}{\mathrm{cvd}}^{\!\textcolor{black}{?}}$};

\node[rectangle, draw, green!60!black, thick, align=center] (td2) at (-1.75,1) {td};

\node[rectangle, draw, green!60!black, thick, align=center] 
(vdp2) at (-4,-0.5) {$\textcolor{green!60!black}{\mathrm{vdp}}^{\!\textcolor{black}{\times}}$};

\node[rectangle, draw, green!60!black, thick, align=center] 
(vds2) at (-2.5,-0.5) {$\textcolor{green!60!black}{\mathrm{vds}}^{\!\textcolor{black}{\times}}$};

\node[rectangle, draw, green!60!black, thick, align=center] 
(fes2) at (-5.5,-0.5) {$\textcolor{green!60!black}{\mathrm{fes}}^{\!\textcolor{black}{?}}$};

\node[rectangle, draw, green!60!black, thick, align=center] 
(vi2) at (-1,-0.5) {$\textcolor{green!60!black}{\mathrm{vi}}^{\!\textcolor{black}{\times}}$};

\node[rectangle, draw, green!60!black, thick, align=center] 
(vdc2) at (1,-0.5) {$\textcolor{green!60!black}{\mathrm{vdc}}^{\!\textcolor{black}{?}}$};

\draw[->] (fes2) -- (fvs2);
\draw[->] (vi2) -- (td2);
\draw[->] (vc2) -- (vi2);
\draw[->] (vc2) -- (tc2);
\draw[->] (nd2) -- (mw2);
\draw[->] (tc2) -- (mw2);
\draw[->] (mw2) -- (cw2);
\draw[->] (tc2) -- (cvd2);
\draw[->] (fvs2) -- (tw2);
\draw[->] (pw2) -- (tw2);
\draw[->] (tw2) -- (cw2);
\draw[->] (cvd2) -- (cw2);
\draw[->] (td2) -- (pw2);

\draw[->] (vc2) -- (nd2);
\draw[->] (vc2) -- (vds2);
\draw[->] (vc2) -- (vdp2);
\draw[->] (vdc2) -- (nd2);
\draw[->] (vdc2) -- (cvd2);
\draw[->] (vdp2) -- (fvs2);
\draw[->] (vds2) -- (td2);
\draw[->] (vds2) -- (fvs2);
\end{scope}
\begin{scope}[yshift=-4.2cm, xshift=4.8cm]

\node[circle, draw, fill=green!60!black, inner sep=2.5pt] at (-6,0) {};
\node[anchor=west] at (-5.6,0) {FPT};

\node[circle, draw, fill=orange!85!black, inner sep=2.5pt] at (-3.7,0) {};
\node[anchor=west] at (-3.3,0) {$\mathrm{W[1]}$-hard};

\node[circle, draw, fill=red!80!black, inner sep=2.5pt] at (-0.8,0) {};
\node[anchor=west] at (-0.4,0) {NP-hard};

\node[circle, draw, fill=black, inner sep=2.5pt] at (1.4,0) {};
\node[anchor=west] at (1.8,0) {open};

\node at (-6,-0.9) {\textcolor{black}{$\star$}};
\node[anchor=west] at (-5.8,-0.9) {Poly kernel};

\node at (-3.7,-0.9) {\textcolor{black}{$\times$}};
\node[anchor=west] at (-3.3,-0.9) {No poly kernel};

\node at (-0.8,-0.9) {\textcolor{black}{$?$}};
\node[anchor=west] at (-0.4,-0.9) {kernel open};

\end{scope}

\end{tikzpicture}
}
\caption{\footnotesize
Complexity landscape of \textsc{Connected Network Microaggregation} under the considered structural parameterizations. 
The left panel shows the complexity when parameterized by the structural parameter alone, while the right panel shows the complexity when the parameter is combined with the cluster size bound $u$. 
An arrow $A \rightarrow B$ indicates that there exists a function $f$ such that $f(A(G)) \geq B(G)$ for every graph $G$. 
Here, $vc$ denotes vertex cover, $nd$ neighborhood diversity, $tc$ twin-cover, $mw$ modular width, $cvd$ cluster vertex deletion, $fvs$ feedback vertex set, $fes$ feedback edge set, $vi$ vertex integrity, $pw$ pathwidth, $tw$ treewidth, $td$ treedepth, $cw$ clique-width, $vdp$ vertex deletion to paths, $vds$ vertex deletion to stars, and $vdc$ vertex deletion to a clique. 
Green indicates fixed-parameter tractability, orange indicates W[1]-hardness, red indicates NP-hardness, and black denotes cases whose status remains open.}
\label{overview}
\end{figure*}
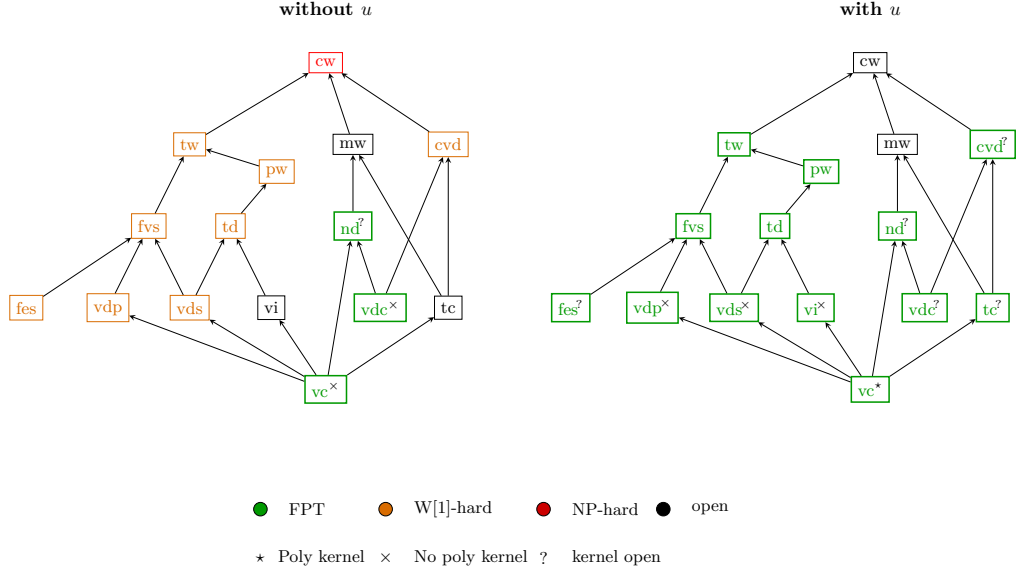

\section{Preliminaries}

We consider simple, undirected graphs $G = (V,E)$ with $n := |V|$ vertices.
For a vertex $v \in V$, we denote its (open) neighborhood by $N(v)$.
For a subset $S \subseteq V$, let $G[S]$ denote the subgraph induced by $S$.
For vertices $u,v \in V$, let $\dist_G(u,v)$ denote the length of a shortest path between $u$ and $v$ in $G$, and $\dist_G(u,v) = \infty$ if no such path exists.
A \emph{clustering} of $V$ is a partition $\Pi = (C_1,\dots,C_m)$ of $V$.
Each set $C_i$ is called a \emph{cluster}, and we associate with each cluster a vertex $c_i \in V$, called its \emph{center}.
In this work, we focus on the \emph{connected} variant, where each cluster $C_i$ must additionally induce a connected subgraph of $G$.

\noindent \textbf{Structural Parameters.}
We now introduce the structural graph parameters used throughout the paper.

\begin{defn}\label{defvc}
A set $S \subseteq V(G)$ is a \emph{vertex cover} of $G$ if every edge in $E(G)$ has at least one endpoint in $S$. The size of a smallest vertex cover of $G$ is called the \emph{vertex cover number}.
\end{defn}

\begin{defn}\label{deffvs}
A \emph{feedback vertex set} of a graph $G$ is a set of vertices whose removal results in a forest. The minimum size of such a set is called the \emph{feedback vertex set number}.
\end{defn}

A rooted forest is a disjoint union of rooted trees. Given a rooted forest $Y$, its \emph{closure} is the graph $H$ with $V(H)=V(Y)$, where two distinct vertices are adjacent if and only if one is an ancestor of the other in $Y$.

\begin{defn}\label{deftd}~\cite{bib12}
The \emph{treedepth} of a graph $G$ is the minimum height of a rooted forest $Y$ whose closure contains $G$ as a subgraph. It is denoted by $td(G)$.
\end{defn}

We also consider the notion of twin-cover, introduced by Ganian~\cite{bib13}.

\begin{defn}~\cite{bib13}
An edge $uv \in E(G)$ is a \emph{twin edge} if $N[u]=N[v]$.
\end{defn}

\begin{defn}~\cite{bib13}\label{deftc}
A set $X \subseteq V(G)$ is a \emph{twin-cover} if every edge of $G$ is either a twin edge or incident to a vertex in $X$. The minimum size of such a set is called the \emph{twin-cover number}, denoted by $tc(G)$.
\end{defn}

Two distinct vertices $u,v$ are called \emph{true twins} if $N[u]=N[v]$ and \emph{false twins} if $N(u)=N(v)$. We say that two vertices have the same \emph{neighborhood type} if they are either true or false twins; such vertices are simply called \emph{twins}.

\begin{defn}~\cite{Lampis}\label{defnd}
A graph $G$ has \emph{neighborhood diversity} at most $d$ if $V(G)$ can be partitioned into at most $d$ sets (called \emph{type classes}) such that all vertices in each set have the same neighborhood type.
\end{defn}

We now recall the notion of tree decompositions introduced by Robertson and Seymour~\cite{Neil}.

\begin{defn}~\cite{Neil}
A \emph{tree decomposition} of a graph $G=(V,E)$ is a tree $T$ together with a family of subsets $(X_t)_{t \in V(T)}$ of $V$ (called \emph{bags}) such that $\bigcup_{t \in V(T)} X_t = V$ and the following conditions hold:
(1) for every edge $uv \in E(G)$, there exists $t \in V(T)$ with $\{u,v\} \subseteq X_t$, and 
(2) for every $v \in V$, the set of nodes $\{t \in V(T) : v \in X_t\}$ induces a connected subtree of $T$.
\end{defn}

\begin{defn}~\cite{Neil}\label{deftw}
The \emph{width} of a tree decomposition is $\max_{t \in V(T)} |X_t| - 1$. The \emph{treewidth} of $G$, denoted $tw(G)$, is the minimum width over all tree decompositions of $G$.
\end{defn}

\begin{defn}\label{defpw}
If the underlying tree $T$ of a tree decomposition is a path, then it is called a \emph{path decomposition}. The \emph{pathwidth} $pw(G)$ is the minimum width over all path decompositions of $G$.
\end{defn}

\begin{defn}\label{cvd def}
The \emph{cluster vertex deletion number} of a graph $G$ is the minimum number of vertices whose removal results in a disjoint union of complete graphs.
\end{defn}

\begin{defn}
    The clique-width of a graph $G$, denoted by ${\tt cw}(G)$, is the minimum number of labels needed 
to construct $G$ 
using the following four operations:
\begin{enumerate}
    \item Create a new graph with a single vertex $v$ with label $i$ (written $i(v)$).
    \item Take the disjoint union of two labelled graphs $G_1$ and $G_2$ (written $G_1 \oplus G_2$).
    \item Add an edge between every vertex with label $i$ and every vertex with label~$j$,
$i\neq j$ (written $\eta_{i,j}$).
\item Relabel every vertex with label $i$ to have label $j$ (written $\rho_{i\rightarrow j}$).
\end{enumerate}
We say that a construction of a graph $G$ with the four operations is a $c$-expression if it 
uses at most $c$ labels. Thus the clique-width of $G$ is the minimum $c$ for which $G$ has a 
$c$-expression.  
A $c$-expression is a rooted binary tree $T$ such that
\begin{enumerate}
    \item each leaf has label $i$ for some $i\in \{1,\ldots,c\}$,
    \item each non-leaf node with two children has label $\oplus $, and 
    \item each non-leaf node with only one child has label $\rho_{i\rightarrow j}$ 
    or $\eta_{i,j}$ $(i,j\in \{1,\ldots,c\}, i\neq j)$.
\end{enumerate} 

\end{defn}

\begin{definition}[Barefoot, Entringer and  Swart.\cite{Barefoot}]\label{defvi} The \emph{vertex integrity} of a graph~$G$, denoted ${\tt vi}(G)$, is the minimum integer $k$ satisfying that there is  $X \subseteq V (G)$ such that~$|X|+|V(C)| \leq  k$ for each component 
$C$ of  $G-X$.
\end{definition}

\medskip
\noindent
We now formally define the problem studied in this paper.

\[
\fbox{
\begin{minipage}{0.92\linewidth}
\textbf{Connected Network Microaggregation (CNMA)} \\[4pt]
\textbf{Input:} 
An undirected graph $G=(V,E)$, and integers $\ell,u,d \in \mathbb{N}$ with $1 \le \ell \le u$. \\[4pt]
\textbf{Question:} 
Does there exist an integer $m$, a partition $\Pi = (C_1,\dots,C_m)$ of $V$, 
and vertices $c_1,\dots,c_m \in V$ such that for every $i \in [m]$:
\begin{itemize}
    \item $\ell \le |C_i| \le u$,
    \item $G[C_i]$ is connected, and
    \item $\dist_G(v,c_i) \le d$ for all $v \in C_i$?
\end{itemize}
Here, the centers $c_i$ are not required to belong to the corresponding clusters $C_i$, and a single vertex may serve as the center for multiple clusters.
\end{minipage}
}
\]

\section{Fixed Parameter Tractability and Kernelization}

\begin{theorem}
\textsc{Connected Network Microaggregation} is fixed-parameter tractable when parameterized by the neighborhood diversity of the input graph.
\end{theorem}

\begin{proof}
Let $(G,d,\ell,u)$ be an instance of \textsc{Connected Network Microaggregation}, and let $r$ be the neighborhood diversity of $G$. We compute in polynomial time a partition $V_1,\dots,V_r$ of $V(G)$ into neighborhood types. Recall that each type induces either a clique or an independent set, and between any two types either all edges are present or none are.

Let $n_i := |V_i|$ for each $i\in [r]$. Define the type graph $H$ with vertex set $[r]$, where $ij \in E(H)$ if and only if every vertex of $V_i$ is adjacent to every vertex of $V_j$ in $G$.
For a cluster $R$, define its support
$S(R) := \{\, i \in [r] \mid R \cap V_i \neq \emptyset \,\}$.
We say that a pair $(S,c)$ with $S \subseteq [r]$ and $c \in [r]$ is \emph{feasible} if the subgraph of $H$ induced by $S$ is connected, and every type in $S$ is at distance at most $d$ from $c$.
Since $r$ is the parameter, all such feasible pairs can be enumerated in polynomial time. Let $\mathcal{F}$ denote the set of all feasible pairs.

\medskip
\noindent
We now construct an integer linear program. For each $(S,c)\in \mathcal{F}$, introduce:
\begin{itemize}
    \item a variable $y_{S,c}$ denoting the number of clusters with support $S$ and center type $c$,
    \item for each $i\in S$, a variable $z_{S,c,i}$ denoting the total number of vertices of type $i$ assigned to such clusters.
\end{itemize}

We restrict our attention to feasible solutions of the ILP.

\[
\fbox{
\begin{minipage}{0.88\linewidth}
\[
\begin{array}{rl}
\text{Subject to} & \\[6pt]

(1) & \displaystyle \sum_{\substack{(S,c)\in \mathcal{F} \\ i\in S}} z_{S,c,i} = n_i
\qquad \text{for all } i \in [r], \\[12pt]

(2) & z_{S,c,i} \ge y_{S,c}
\qquad \text{for all } (S,c)\in \mathcal{F},\ i\in S, \\[12pt]

(3) & \ell \cdot y_{S,c} \le \displaystyle\sum_{i\in S} z_{S,c,i} \le u \cdot y_{S,c}
\qquad \text{for all } (S,c)\in \mathcal{F}, \\[12pt]

(4) & z_{S,c,i},\ y_{S,c} \in \mathbb{Z}_{\ge 0}.
\end{array}
\]
\end{minipage}
}
\]

\medskip
\noindent
\textbf{Explanation of the constraints.}
Constraint~(1) ensures that all vertices are covered: for every type $i \in [r]$, the total number of vertices assigned across all clusters equals $n_i$.
Constraint~(2) enforces consistency between cluster selection and assignments: if a cluster of type $(S,c)$ is selected (i.e., $y_{S,c} \ge 1$), then each type $i \in S$ must contribute at least one vertex to that cluster.
Constraint~(3) guarantees that every selected cluster has a valid size, namely between $\ell$ and $u$.
Finally, Constraint~(4) enforces integrality and nonnegativity of all variables.

\medskip
\noindent
\textbf{Correctness.}
Suppose first that a feasible clustering exists. For each cluster $R$, let $S=S(R)$ and let $c$ be any valid center. Then $(S,c)\in \mathcal{F}$. Set $y_{S,c}$ to be the number of such clusters, and $z_{S,c,i}$ the number of vertices of type $i$ in them. All constraints are satisfied.

Conversely, suppose the ILP has a feasible solution. Fix $(S,c)\in \mathcal{F}$ and let $x := y_{S,c}$ and $M_i := z_{S,c,i}$ for each $i\in S$. Let
$M := \sum_{i\in S} M_i$.
Since $M_i \ge x$, assign one vertex of type $i$ to each of the $x$ clusters. Each cluster now contains all types in $S$ and has size $|S|$.
Since $\ell x \le M \le ux$, there exist integers $m_1,\dots,m_x$ such that
\[
\ell \le m_j \le u \quad \text{for all } j\in [x],
\qquad
\sum_{j=1}^x m_j = M.
\]
This follows by writing $M = \ell x + t$ with $0 \le t \le (u-\ell)x$ and distributing $t$ units among the $x$ clusters.
Distribute the remaining vertices arbitrarily so that cluster $j$ reaches size exactly $m_j$.
Each resulting cluster is connected since $S$ induces a connected subgraph of the type graph and each cluster contains at least one vertex from each type in $S$. Moreover, since $(S,c)\in \mathcal{F}$, there exists a vertex of type $c$ at distance at most $d$ from all vertices of types in $S$, hence every cluster satisfies the distance constraint.
Applying this construction for all $(S,c)\in \mathcal{F}$ yields a valid partition of $V(G)$.

\medskip
\noindent
\textbf{Complexity.}
There are at most $r \cdot 2^r$ feasible pairs $(S,c)$. For each such pair, we introduce one variable $y_{S,c}$ and at most $r$ variables $z_{S,c,i}$. Thus the total number of variables is at most $(r+1)2^r$. By Lenstra's theorem, the ILP can be solved in time depending only on $r$. Hence the problem is fixed-parameter tractable parameterized by neighborhood diversity.
\end{proof}

\begin{corollary}
    \textsc{Connected Network Microaggregation} is FPT when parameterized by vertex deletion to clique.
\end{corollary}

\begin{theorem}
\textsc{Connected Network Microaggregation} admits a polynomial kernel of size $\mathrm{vc} \cdot u$ when parameterized by $\mathrm{vc}+u$.
\end{theorem}

\begin{proof}
Let $(G,d,\ell,u)$ be an instance of \textsc{Connected Network Microaggregation}, and let $X$ be a vertex cover of $G$ of size $k:=|X|$.

If $\ell=1$, then the instance is trivially a yes-instance: indeed, we may partition $V(G)$ into singleton clusters, each of which is connected and has size in $[\ell,u]$, and every singleton admits its unique vertex as a center. Hence we may assume from now on that $\ell>1$.

Under this assumption, every valid cluster has size at least $2$. Since $V(G)\setminus X$ is an independent set, no connected induced subgraph of size at least $2$ can be contained entirely in $V(G)\setminus X$. Therefore every cluster in any feasible solution must contain at least one vertex of the vertex cover $X$.

This immediately implies that any feasible solution consists of at most $k$ clusters, because distinct clusters are vertex-disjoint and each contains at least one vertex of $X$. As every cluster has size at most $u$, the total number of vertices in a yes-instance is at most $ku$.
We therefore apply the following reduction rule.

\begin{cma}
If $|V(G)|>ku$, return a fixed no-instance.
\end{cma}

\noindent
The rule is safe. Indeed, if $|V(G)|>ku$, then no partition of $V(G)$ into at most $k$ clusters of size at most $u$ exists, and hence the instance cannot be a yes-instance.
If the rule is not applicable, then the graph already has at most $ku$ vertices. Thus the instance itself is a kernel of size at most $ku$.
Since the rule can clearly be checked and applied in polynomial time, this yields a polynomial kernel for \textsc{Connected Network Microaggregation} parameterized by $\mathrm{vc}+u$.
\end{proof}

\begin{theorem}\label{thm:connected-dp-tw-u}
\textsc{Connected Network Microaggregation} can be solved in time $(u(w+1)^2)^{2(w+1)} \cdot n^{O(1)}$, where $w$ is the treewidth of the input graph, provided that $d>u$.
\end{theorem}

\begin{proof}
Observe that if $d>u$, then every connected graph on at most $u$ vertices has diameter at most $u-1<d$. Hence, the distance constraint is automatically satisfied, and it suffices to partition $V(G)$ into connected components of size between $\ell$ and $u$.

Let $(T,\{X_t\}_{t\in V(T)})$ be a nice tree decomposition of width $w$. For a node $t$, let $G_t$ denote the subgraph induced by vertices appearing in the subtree rooted at $t$.
We define a dynamic program. For each node $t$, we store entries $\mathrm{DP}[t,\mathcal S]\in\{\textsf{true},\textsf{false}\}$, where a state $\mathcal S$ consists of:
\begin{itemize}
    \item a partition $\Pi$ of $X_t$;
    \item a function $\sigma:\Pi\to \{1,\dots,u\}$, where $\sigma(B)$ denotes the size of the partial cluster corresponding to block $B$;
    \item for each block $B\in\Pi$, a partition $\Gamma_B$ of $B$ describing connectivity among vertices of $B$ inside $G_t$.
\end{itemize}

The interpretation is that each block $B\in\Pi$ corresponds to one partial cluster intersecting the bag, with size $\sigma(B)$, and $\Gamma_B$ records how vertices of $B$ are connected within $G_t$.

\medskip
\noindent
\textbf{Leaf node.} For a leaf node with empty bag, we set
\[
\mathrm{DP}[t,(\emptyset,\emptyset,\emptyset)] = \textsf{true}.
\]

\medskip
\noindent
\textbf{Introduce node.} Suppose $t$ introduces a vertex $v$, and let $t'$ be its child. For each state $\mathcal S'$ at $t'$, we construct states at $t$ as follows:

\begin{itemize}
    \item $v$ starts a new block $\{v\}$ with $\sigma(\{v\})=1$ and $\Gamma_{\{v\}}=\{\{v\}\}$;
    \item or $v$ joins an existing block $B\in\Pi'$, provided $\sigma'(B)<u$. In this case, we update $\sigma(B)=\sigma'(B)+1$ and update $\Gamma_B$ by merging all components containing neighbors of $v$ in $B$, or placing $v$ as a singleton if it has no neighbors in $B$.
\end{itemize}

\medskip
\noindent
\textbf{Forget node.} Suppose $t$ forgets a vertex $v$, and let $t'$ be its child. Let $B$ be the block containing $v$.

\begin{itemize}
    \item If $B=\{v\}$, then this cluster is closed, and we require $\sigma(B)\in[\ell,u]$ and $\Gamma_B=\{\{v\}\}$; otherwise the transition is invalid.
    \item If $|B|\ge 2$, we remove $v$ from $B$ and restrict $\Gamma_B$ accordingly, keeping $\sigma(B)$ unchanged.
\end{itemize}

\medskip
\noindent
\textbf{Join node.} Suppose $t$ has children $t_1,t_2$ with $X_t=X_{t_1}=X_{t_2}$. We combine states $\mathcal S_1,\mathcal S_2$ with identical partitions $\Pi$.
For each block $B\in\Pi$, we set
$\sigma(B) = \sigma_1(B) + \sigma_2(B) - |B|$,
and require $\sigma(B)\le u$.
The connectivity partition $\Gamma_B$ is obtained as the transitive closure of the union of $\Gamma_B^1$ and $\Gamma_B^2$.

\medskip
\noindent
\textbf{Root.} At the root with empty bag, we accept if
\[
\mathrm{DP}[r,(\emptyset,\emptyset,\emptyset)] = \textsf{true}.
\]

\medskip
\noindent
\textbf{Complexity.} Let $b=w+1$. The number of partitions of a bag $X_t$ is at most $\mathrm{Bell}(b) \le b^b$. For each partition $\Pi$, assigning sizes contributes at most $u^{|\Pi|} \le u^b$ possibilities. For connectivity, for each block $B \in \Pi$, the number of partitions $\Gamma_B$ is at most $\mathrm{Bell}(|B|) \le b^{|B|}$, and hence over all blocks this contributes at most $b^b$. Thus, the total number of states per bag is at most $u^b \cdot b^{2b} = u^{w+1}(w+1)^{2(w+1)}$.
Each introduce and forget operation can be processed in time polynomial in $b$, while a join node requires combining pairs of states, leading to $O(N_b^2 \cdot b^{O(1)})$ time per node, where $N_b = u^{w+1}(w+1)^{2(w+1)}$. Since a nice tree decomposition has $O(wn)$ nodes, the total running time is $(u(w+1)^2)^{2(w+1)} \cdot n^{O(1)}$.
\end{proof}

\begin{corollary}
\textsc{Connected Network Microaggregation} is fixed-parameter tractable when parameterized by $\mathrm{tw}+u$.
\end{corollary}

\begin{proof}
Let $(G,d,\ell,u)$ be an instance of \textsc{Connected Network Microaggregation}. If $d \leq u$, then the problem is fixed-parameter tractable when parameterized by $\mathrm{tw}+d+u$~\cite{Blažej_Ganian_Knop_Pokorný_Schierreich_Simonov_2023}, and hence also when parameterized by $\mathrm{tw}+u$.
Otherwise, if $d>u$, then by Theorem~\ref{thm:connected-dp-tw-u} the problem can be solved in time $(u(\mathrm{tw}+1)^2)^{2(\mathrm{tw}+1)} \cdot n^{O(1)}$, which is fixed-parameter tractable with respect to $\mathrm{tw}+u$.
Thus, in all cases, \textsc{Connected Network Microaggregation} is fixed-parameter tractable when parameterized by $\mathrm{tw}+u$.
\end{proof}

\begin{theorem}
\textsc{Connected Network Microaggregation} is fixed-parameter tractable when parameterized by the cluster vertex deletion number of the input graph plus the upper bound $u$.
\end{theorem}

\begin{proof}
Let $(G,d,\ell,u)$ be an instance of \textsc{Network Microaggregation}, and let $X$ be a cluster vertex deletion set of $G$ of size $k$, that is, $G-X$ is a cluster graph. Let $C_1,\dots,C_t$ denote the cliques of $G-X$. 
For every vertex $v \in V(G)\setminus X$, define its type as the set $N(v)\cap X$. Hence there are at most $2^k$ possible types. Let $R_1,\dots,R_{k'}$ denote the clusters that intersect $X$. Since these clusters induce a partition of $X$, the number $k'$ is at most $k$. We first guess this partition of $X$; the number of possibilities is at most $k^k$.

We now guess the structure of each cluster $R_i$. Observe that $R_i$ consists of the vertices in $R_i\cap X$ together with some vertices of $G-X$. Since $G-X$ is a cluster graph, the vertices of $R_i\setminus X$ induce a disjoint union of cliques. Moreover, since every cluster has size at most $u$, the total number of vertices in $R_i$ is at most $u$, and hence $R_i\setminus X$ is a disjoint union of at most $u$ cliques, each of size at most $u$.

For each of these clique-components, we guess how many vertices of each type it contains. More precisely, if a clique-component $Q$ contributes to $R_i$, then for every type $A \subseteq X$ we guess the number $\eta_Q(A)$ of vertices in $Q$ whose neighborhood in $X$ is exactly $A$. Since there are at most $2^k$ possible types and each value $\eta_Q(A)$ lies in $\{0,\dots,u\}$, the number of possibilities for one clique-component is at most $(u+1)^{2^k}$.
Now $R_i\setminus X$ is a disjoint union of at most $u$ clique-components, because $|R_i|\leq u$. Hence the number of possibilities for the whole part $R_i\setminus X$ is at most $((u+1)^{2^k})^u = (u+1)^{u\cdot 2^k}$.
It remains to choose the set $R_i\cap X$. Since $R_i\cap X$ is a subset of $X$, there are at most $2^k$ possibilities for this choice. Therefore, the total number of possibilities for one cluster $R_i$ is at most $2^k \cdot (u+1)^{u\cdot 2^k}$.
Since the number of clusters intersecting $X$ is at most $k$, the total number of guesses for the structures of all clusters $R_1,\dots,R_{k'}$ is at most $(2^k \cdot (u+1)^{u\cdot 2^k})^k = 2^{k^2}\cdot (u+1)^{k u\cdot 2^k}$. Thus the total number of such guesses is bounded explicitly by $2^{k^2}\cdot (u+1)^{k u\cdot 2^k}$.

For each guessed cluster $R_i$, we also guess its center. If the center lies in $X$, then there are at most $k$ possibilities. Otherwise, the center lies in $G-X$. In this case, we describe the center by the pair $(N(v)\cap X,\; N(N(v))\cap X)$ together with the following additional information: whether $v$ lies in one of the guessed clique-components contributing to $R_i$, and if so, in which clique-component of $R_i$ it lies.
Since the structure of $R_i$ has already been guessed, this information suffices to determine the distances from $v$ to all vertices of $R_i$, and hence to verify whether $v$ can serve as a valid center of $R_i$. In particular, this works already for any $d$, and we do not need to assume $d\geq 3$.
The number of possibilities for the pair $(N(v)\cap X,\; N(N(v))\cap X)$ is at most $4^k$. Moreover, the number of guessed clique-components of $R_i$ is at most $u$, since $|R_i|\leq u$. Hence the additional information can be chosen in at most $u+1$ ways, where the extra possibility corresponds to the case that $v$ does not lie in any of the clique-components of $R_i$. Therefore, for each cluster $R_i$, the center can be guessed in at most $k + (u+1)\cdot 4^k$ ways.
Not all guesses constructed above correspond to valid clusters. We therefore discard every guess for which the following conditions are not satisfied.
\begin{itemize}
    \item The graph $G[R_i]$ is connected.
    \item There exists a vertex $v$ (not necessarily in $R_i$) consistent with the guessed description $(N(v)\cap X,\; N(N(v))\cap X)$ such that the distance from $v$ to every vertex of $R_i$ is at most $d$.
\end{itemize}
Since the structure of $R_i$ and the description of $v$ have already been guessed, both conditions can be verified in time depending only on $k+u$. In particular, the distances from $v$ to all vertices of $R_i$ are determined by the guessed types and the placement of $v$ relative to the clique-components.
We henceforth assume that all remaining guesses satisfy these properties.

After these guesses, the clusters intersecting $X$ are fixed up to the actual realization of the guessed clique-components inside the cliques $C_1,\dots,C_t$ of $G-X$. The remaining task is to determine, for each clique $C_j$, which of its vertices are used to realize one of the guessed components of some $R_i$, and how the remaining vertices of $C_j$ are grouped into clusters disjoint from $X$.
To solve this, we use the natural path decomposition of $G$ whose bags are $B_j := X \cup C_j$ for $j \in [t]$, ordered arbitrarily. This is indeed a path decomposition. Since each bag contains all of $X$ and exactly one clique of $G-X$, the interaction between different bags is mediated solely through the bounded set $X$. Consequently, once the clusters intersecting $X$ and their centers have been guessed, the remaining choices can be handled by dynamic programming over this path decomposition with a state space bounded by a function of $k+u$. It remains to define the dynamic programming states and transitions.

\medskip
\noindent
\medskip
\noindent
\textbf{Dynamic Programming over the Path Decomposition.}

\noindent Recall that we process the cliques $C_1,\dots,C_t$ of $G-X$ in order, where $B_j := X \cup C_j$.
For each cluster $R_i$, we have already guessed:
\begin{itemize}
    \item the set $R_i \cap X$,
    \item clique-components $P^i_1,\dots,P^i_{u_i}$ of $R_i \setminus X$,
    \item for each $P^i_s$ and type $A \subseteq X$, the number $\eta_{P^i_s}(A)$,
    \item and the description of a valid center via $(N(v)\cap X,\; N(N(v))\cap X)$.
\end{itemize}

\noindent Importantly, the center is \emph{not required} to belong to $R_i$, nor to any $P^i_s$.

\medskip

\noindent
\textbf{DP State.}
For $j \in \{0,\dots,t\}$, we define:
\[
\mathrm{DP}[j,\mathcal{S},\mathcal{C}] \in \{\texttt{true},\texttt{false}\},
\]
where:
\begin{itemize}
    \item $\mathcal{S} \subseteq \{(i,s) \mid i \in [k'],\, s \in [u_i]\}$ is the set of components that have been realized in the cliques $C_1,\dots,C_j$,
    \item $\mathcal{C} \subseteq [k']$ is the set of indices $i$ such that a valid center vertex for $R_i$ has already been encountered in one of the cliques $C_1,\dots,C_j$.
\end{itemize}

\medskip

\noindent
\textbf{Initialization.}
$\mathrm{DP}[0,\emptyset,\emptyset] = \texttt{true}$, and all other entries are \texttt{false}.

\medskip

\noindent
\textbf{Transitions.}
For each $j \in [t]$, we compute $\mathrm{DP}[j,\cdot,\cdot]$ from $\mathrm{DP}[j-1,\cdot,\cdot]$.
Fix a state $(\mathcal{S},\mathcal{C})$ with $\mathrm{DP}[j-1,\mathcal{S},\mathcal{C}] = \texttt{true}$.
Let $n_j(A)$ denote the number of vertices of type $A \subseteq X$ in $C_j$.

\medskip
\noindent
\emph{Choosing realized components.}
We consider all sets $T \subseteq \{(i,s) \mid (i,s) \notin \mathcal{S}\}$
such that for every $i \in [k']$, there is \emph{at most one} index $s$ with $(i,s) \in T$.
Equivalently, for each cluster $R_i$, we choose either no component or exactly one component $P^i_s$ (not yet in $\mathcal{S}$) to be realized in the clique $C_j$.

\medskip

\noindent
\emph{Feasibility condition.}
For every type $A \subseteq X$, we require:
\[
\sum_{(i,s)\in T} \eta_{P^i_s}(A) \;\le\; n_j(A).
\]

\medskip

\noindent
\emph{Leftover condition.}
Let
\[
L_j := \sum_{A \subseteq X} \left( n_j(A) - \sum_{(i,s)\in T} \eta_{P^i_s}(A) \right).
\]
We require that there exists an integer $a \ge 0$ such that $a\ell \;\le\; L_j \;\le\; a u$.

\medskip

\noindent
\emph{Center realization.}
For each cluster $R_i \notin \mathcal{C}$, we check whether the current clique $C_j$ contains a vertex that is compatible with the guessed center description of $R_i$.
Formally, this holds if there exists a type $A \subseteq X$ such that:
\begin{itemize}
    \item $n_j(A) > 0$, and
    \item $A$ is consistent with $(N(v)\cap X,\; N(N(v))\cap X)$.
\end{itemize}
Let $\mathcal{C}'$ be obtained from $\mathcal{C}$ by adding all such indices $i$.

\medskip

\noindent
\emph{State update.}
We set:
\[
\mathrm{DP}[j,\mathcal{S} \cup T,\mathcal{C}'] = \texttt{true}.
\]

\medskip

\noindent
\textbf{Acceptance.}
We accept if there exists a state $(\mathcal{S},\mathcal{C})$ such that:
\begin{itemize}
    \item $\mathrm{DP}[t,\mathcal{S},\mathcal{C}] = \texttt{true}$,
    \item $\mathcal{S} = \{(i,s) \mid i \in [k'],\, s \in [u_i]\}$,
    \item $\mathcal{C} = [k']$.
\end{itemize}

\medskip

\noindent
\textbf{Correctness.}
Each component $P^i_s$ is realized entirely within a single clique $C_j$, which is valid since $G-X$ is a cluster graph. The leftover condition ensures that vertices not assigned to any $R_i$ can be partitioned into clusters of sizes in $[\ell,u]$. The center condition guarantees that for each cluster $R_i$, there exists a vertex (possibly outside $R_i$) satisfying the required neighborhood constraints. Since the structure of each $R_i$ was guessed to be connected, realizing all its components suffices to ensure connectivity.

\medskip

\noindent
\textbf{Running Time.}
The number of components over all clusters is at most $ku$. Hence the number
of DP states is at most $ 2^{ku} \cdot 2^k  \cdot n$.
For each state, we consider at most $(u+1)^k$ choices for the set $T$, since
for each cluster $R_i$ we either select no component or one of its at most $u$
components. Each transition can be evaluated in time $2^k.n^{\mathcal{O}(1)}$, since all checks are performed over the set of types $A \subseteq X$, of which there are at most $2^k$.
In particular, verifying the feasibility constraints, computing the leftover,
and checking the center condition can all be done by iterating over these types.
Thus, the total running time of the dynamic programming is
$ 2^{O(ku)} \cdot (u+1)^k \cdot n^{\mathcal{O}(1)}$.
Combining this with the number of guesses, the total running time is
$2^{O(k^2 + ku)} \cdot (u+1)^{ku \cdot 2^k + O(k)}\cdot n^{O(1)}$.
\end{proof}

Next, we show that the problem is fixed-parameter tractable when parameterized by the cluster vertex deletion number together with the size of the largest clique. This result requires $d \geq 3$; the complexity of the problem for $d \in \{1,2\}$ remains open.

\begin{theorem}
\textsc{Connected Network Microaggregation} for any $d\geq 3$ is fixed-parameter tractable when parameterized by 
$k + \omega$, where $k$ is the cluster vertex deletion number of the input graph and 
$\omega$ is the maximum clique size in $G-X$, where $X$ is a cluster vertex deletion set.
\end{theorem}
\begin{proof}
Let $X \subseteq V(G)$ be a cluster vertex deletion set such that $G-X$ is a cluster graph. 
We begin by guessing the partition of the vertices of $X$ into clusters. 
More precisely, we guess all clusters that intersect $X$, together with their intersections with $X$. 
Since these intersections form a partition of $X$, the number of such guesses is bounded by $k^{\mathcal{O}(k)}$.
For each guessed cluster intersecting $X$, we also guess its center. 
If the center lies in $X$, then there are at most $k$ possible choices. 
Otherwise, the center lies in $V(G)\setminus X$. 
In this case, the following Claim~\ref{unique center} proves that the behavior of the center is completely determined by the pair $(N(c)\cap X,\; N(N(c))\cap X)$.

\begin{claim}\label{unique center}
Let $X$ be a cluster vertex deletion set of $G$, and suppose $d\ge 3$. 
Let $a,b\in V(G)\setminus X$ be vertices belonging to (not necessarily distinct) cliques of $G-X$ such that
$N(a)\cap X = N(b)\cap X$ and $N(N(a))\cap X = N(N(b))\cap X$.
Then for every vertex $u\in V(G)$ it holds that
$\operatorname{dist}_G(u,a)\le d \quad\text{if and only if}\quad \operatorname{dist}_G(u,b)\le d$.
\end{claim}

\begin{proof}
Since $G-X$ is a cluster graph, every vertex of $V(G)\setminus X$ lies in a clique of $G-X$. 
Consider any vertex $u\in V(G)$.
If $u\in X$, then $u$ is adjacent to $a$ if and only if $u\in N(a)\cap X$, which holds if and only if $u\in N(b)\cap X$. 
Hence $\operatorname{dist}_G(u,a)=1 \quad\text{if and only if}\quad \operatorname{dist}_G(u,b)=1$.
If $u\notin X$, then $u$ lies in some clique of $G-X$. Since $d\ge 2$, the only paths of length at most two from $u$ to $a$ or $b$ must pass through vertices of $X$ or through vertices in the same clique. 
The vertices of $X$ that are adjacent to $a$ and those at distance two from $a$ are exactly the sets $N(a)\cap X$ and $N(N(a))\cap X$, respectively. 
By assumption these sets coincide with $N(b)\cap X$ and $N(N(b))\cap X$.
Therefore the same vertices of $X$ serve as intermediaries for paths from $u$ to $a$ and from $u$ to $b$ of length at most two.
Consequently $u$ is at distance at most $d$ from $a$ if and only if it is at distance at most $d$ from $b$. 
\end{proof}

\noindent \textbf{Remark.} We note that the assumption $d \ge 3$ is necessary for the above argument. 
Indeed, the statement of Claim~\ref{unique center} fails for $d=2$.
To see this, consider the following example. 
Let $a$ and $b$ be vertices in two distinct cliques of $G-X$ such that  $N(a)\cap X = N(b)\cap X$ and  $N(N(a))\cap X = N(N(b))\cap X$.
Let $u_1$ be a vertex adjacent to $a$ but not to $b$, and let $u_2$ be a vertex adjacent to $b$ but not to $a$, where both $u_1$ and $u_2$ lie in cliques of $G-X$.
Then we have
$\operatorname{dist}(u_1,a)=1, \quad \operatorname{dist}(u_1,b)=3$,
and
$\operatorname{dist}(u_2,b)=1, \quad \operatorname{dist}(u_2,a)=3$.
Thus, for $d=2$, the vertices $a$ and $b$ are not equivalent with respect to distance-$d$ neighborhoods, even though they agree on $N(\cdot)\cap X$ and $N(N(\cdot))\cap X$.
This shows that for $d=2$, the pair $(N(c)\cap X,\; N(N(c))\cap X)$ does not uniquely determine the center behavior, and hence our approach requires the assumption $d \ge 3$. \\

Since each of these sets is a subset of $X$, there are at most $2^k\cdot 2^k = 4^k$ possibilities for such a pair. 
Hence, for every cluster intersecting $X$, the center can be guessed in at most $k+4^k$ ways.
Since every cluster intersecting $X$ contains at least one vertex of $X$, 
the number of such clusters is at most $k$. 
The number of ways to partition $X$ into at most $k$ parts is at most $k^k$. 
For each cluster, the center can be guessed in at most $k+4^k$ ways. 
Therefore, the total number of guesses is at most $k^k \cdot (k+4^k)^k \;\le\; (k+4^k)^k \cdot k^k \;\le\; (4^k + k)^{2k}$.

For every vertex $v \in V(G)\setminus X$, define its \emph{type} as the set 
$N(v)\cap X$. For each subset $A \subseteq X$, let 
$V_A := \{ v \in V(G)\setminus X \mid N(v)\cap X = A \}$.
Thus, the vertices of $G-X$ are partitioned into at most $2^k$ types.
Let $R_1,\dots,R_{k'}$ denote the clusters that intersect $X$, where $k' \le k$. 
For a clique $C$ of $G-X$, we refine the above description by considering how the vertices of $C$ are distributed among the clusters $R_i$ according to their types. 
For every $i \in [k']$ and every $A \subseteq X$, define
$\eta_C(i,A) := |R_i \cap V_A \cap C|$.
That is, $\eta_C(i,A)$ counts the number of vertices of type $A$ in the clique $C$ that are assigned to the cluster $R_i$.
We say that two cliques $C$ and $C'$ of $G-X$ are \emph{equivalent} if
$\eta_C(i,A) = \eta_{C'}(i,A) \quad \text{for all } i \in [k'] \text{ and all } A \subseteq X$ and also both cliques are of equal size.
In other words, two cliques are equivalent if their vertices are distributed in exactly the same way across the clusters $R_i$, with respect to all types.
Observe that each clique $C$ can thus be represented by a vector 
$\big( \eta_C(i,A) \big)_{i \in [k'],\, A \subseteq X}$,
which has at most $k \cdot 2^k$ entries. Each entry is a non-negative integer bounded by $|C|$, and hence by $\omega := \max\{|C| : C \text{ is a clique of } G-X\}$.
Therefore, the number of equivalence classes of cliques is at most
$(\omega+1)^{(k \cdot 2^k +1)}$. \\

Let $\mathcal{P}$ denote the set of all patterns (equivalence classes). 
Each pattern $p \in \mathcal{P}$ consists of:
\begin{itemize}
    \item an integer $s(p) \in \{0,\dots,\omega\}$ denoting the size of the clique, and
    \item integers $\eta_p(i,A)$ for every $i \in [k']$ and $A \subseteq X$, where $\eta_p(i,A) \in \{0,\dots,\omega\}$.
\end{itemize}

We now determine which of these patterns can actually appear in a feasible solution. 
Recall that we have already fixed a partition of $X$ into clusters $R_1,\dots,R_{k'}$ together with a guessed center $c_i$ for each cluster $R_i$.

For each $i \in [k']$ and each type $A \subseteq X$, we test whether vertices of type $A$ can be assigned to $R_i$. 
More precisely, we check whether there exists a vertex $v \in V_A$ such that $\operatorname{dist}_G(v,c_i) \le d$. 
By Claim~\ref{unique center}, this condition depends only on the type $A$ and the chosen center, and hence can be checked in polynomial time.

If no such vertex exists, then no vertex of type $A$ can belong to $R_i$ in any feasible clustering. 
Therefore, for every pattern $p$, we must have $\eta_p(i,A)=0$.

We impose an additional feasibility condition on each pattern. 
For a pattern $p$, consider the number of vertices of a clique that are not assigned to any of the clusters $R_1,\dots,R_{k'}$, that is,
$s(p) - \sum_{i \in [k']} \sum_{A \subseteq X} \eta_p(i,A)$.
These vertices must be covered by clusters that are disjoint from $X$. 
Therefore, this quantity must be expressible as a sum of integers from the interval $[\ell,u]$. 
This condition can be checked in polynomial time, and if it is not satisfied, the pattern is discarded.

We call a pattern $p \in \mathcal{P}$ \emph{feasible} if it satisfies both conditions: 
(i) $\eta_p(i,A)=0$ for all inadmissible pairs $(i,A)$, and 
(ii) the remaining vertices $s(p) - \sum_{i,A} \eta_p(i,A)$ can be written as a sum of integers from $[\ell,u]$.
Thus, after fixing the partition of $X$ and the centers of the clusters intersecting $X$, we can in polynomial time filter the set $\mathcal{P}$ and retain only those patterns that are feasible.
We now ensure that all vertices of $G-X$ are covered. 
Recall that for each pattern $p$, the variable $x_p$ denotes the number of cliques of type $p$, and each such clique contributes $s(p)$ vertices.
Therefore, we impose the constraint
\[
\sum_{p \in \mathcal{P}} s(p)\, x_p \;=\; |V(G)\setminus X|.
\]

We now enforce the size constraints for clusters intersecting $X$. 
Recall that for each $i \in [k']$, the set $R_i \cap X$ is fixed by the initial guess, and hence its size is known.
The remaining vertices of $R_i$ come from cliques of $G-X$ according to the patterns. 
For each pattern $p$, the value $\sum_{A \subseteq X} \eta_p(i,A)$ denotes the number of vertices contributed to $R_i$ by a single clique of type $p$. 
Since there are $x_p$ such cliques, the total contribution to $R_i$ is 
$\sum_{p \in \mathcal{P}} x_p \cdot \sum_{A \subseteq X} \eta_p(i,A)$.
Thus, for every $i \in [k']$, we impose the constraint
\[
\ell \;\le\; |R_i \cap X| \;+\; \sum_{p \in \mathcal{P}} x_p \cdot \sum_{A \subseteq X} \eta_p(i,A) \;\le\; u.
\]

The number of variables is bounded by $|\mathcal{P}| \le (\omega+1)^{k\cdot 2^k + 1}$, 
and hence depends only on the parameter. Therefore, by Lenstra's theorem, 
the ILP can be solved in fixed-parameter tractable time.
\end{proof}

\section{Lower bounds on Kernelization}

\begin{theorem}
Unless coNP $\subseteq$ NP/poly, \textsc{Connected Network Microaggregation}
does not admit a polynomial kernel when parameterized by the vertex cover number
of the input graph, even for all $d \ge 3$.
\end{theorem}

\begin{proof}
We give a polynomial parameter transformation from
\textsc{Red-Blue Dominating Set} parameterized by $|B|$.
Recall that an instance of \textsc{Red-Blue Dominating Set} consists of a bipartite graph
$G=(R \uplus B,E)$ and an integer $k$, and the task is to decide whether there exists
a set $S \subseteq R$ of size at most $k$ such that every vertex of $B$ has a neighbor in $S$.
It is known that this problem does not admit a polynomial kernel when parameterized by $|B|+k$,
unless coNP $\subseteq$ NP/poly.

Let $(G=(R \uplus B,E),k)$ be an instance of \textsc{Red-Blue Dominating Set}.
Set $n:=|R|$ and $t:=|B|$. We may assume that $k \le n$ and $t \ge 1$, since otherwise
the instance is trivial.
We construct an instance $(G',d,\ell,u)$ of \textsc{Connected Network Microaggregation} as follows.
We keep the graph $G$.
For every blue vertex $z \in B$, we add a set $P_z$ of $c$ pendant vertices adjacent only to $z$.
Next, we add two new vertices $x$ and $y$, each adjacent to all vertices of $R$.

We now attach gadgets to $x$ and $y$.
For $x$, add $a-6$ pendant vertices adjacent only to $x$, and two internally vertex-disjoint
paths of length $3$ starting at $x$, say
$x-x_1-x_2-x_3$
and
$x-x'_1-x'_2-x'_3$.
Thus the gadget of $x$ contributes exactly $a$ vertices besides $x$.
Similarly, for $y$, add $b-6$ pendant vertices adjacent only to $y$, and two internally vertex-disjoint
paths of length $3$ starting at $y$, say
$y-y_1-y_2-y_3$
and
$y-y'_1-y'_2-y'_3$.
Thus the gadget of $y$ contributes exactly $b$ vertices besides $y$.
We define
$c := n-k+6,\qquad
b := c+1 = n-k+7,\qquad
a := k+7+t(c+1).$
Finally, let $u := 1+a+(n-k)$. Note that then
$u = 1+b+k+t(1+c)$, and furthermore
$1+a+1+c > u$ and $a>b>c$.
Finally, let $d \ge 3$ be arbitrary, and set $\ell := u$.
By construction,
\[
|V(G')|
= n + t + tc + 2 + a + b
= \bigl(1+a+(n-k)\bigr) + \bigl(1+b+k+t(1+c)\bigr)
= 2u.
\]
Hence every feasible clustering of $G'$ consists of exactly two clusters, each of size exactly $u$.

\medskip
\noindent
\textbf{Parameter bound.}
The graph $G'$ has a vertex cover of size at most $|B|+6$, namely
$B \cup \{x,y,x_2,x'_2,y_2,y'_2\}$.
Indeed, the vertices of $B$ cover all edges of the original bipartite graph and all edges to the pendant vertices $P_z$;
the vertex $x$ covers all edges from $x$ to $R$ and all pendant edges of the $x$-gadget;
the vertices $x_2,x'_2$ cover the remaining edges on the two length-$3$ paths of the $x$-gadget;
and analogously $y,y_2,y'_2$ cover all edges incident with the $y$-gadget and the edges from $y$ to $R$.
Hence the parameter of the produced instance is bounded polynomially in $|B|$.

\medskip
\noindent
\textbf{Forward direction.}
Assume that $(G,k)$ is a yes-instance of \textsc{Red-Blue Dominating Set}.
Then there exists a set $S \subseteq R$ of size at most $k$ dominating $B$.
By adding arbitrary vertices of $R \setminus S$ if necessary, we may assume that $|S|=k$.
We define
$C_x := \{x\} \cup V(X\text{-gadget}) \cup (R \setminus S)$ and
$C_y := \{y\} \cup V(Y\text{-gadget}) \cup S \cup B \cup \bigcup_{z\in B} P_z$.

By the choice of $u$,
$|C_x| = 1+a+(n-k)=u$,
and
$|C_y| = 1+b+k+t(1+c)=u$.
Thus $C_x$ and $C_y$ form a partition of $V(G')$ into two sets of size exactly $u$.
We claim that both $G'[C_x]$ and $G'[C_y]$ are connected and admit a center within distance at most $3$.
For $C_x$, the graph $G'[C_x]$ is connected because every vertex of the gadget of $x$ is connected to $x$, and every vertex of $R \setminus S$ is adjacent to $x$.
Moreover, $x$ is a center: every vertex of the gadget of $x$ is at distance at most $3$ from $x$, and every vertex of $R \setminus S$ is adjacent to $x$.

For $C_y$, the graph $G'[C_y]$ is connected because the gadget of $y$ is connected to $y$, every vertex of $S$ is adjacent to $y$, every blue vertex $z \in B$ has a neighbor in $S$ by domination, and every pendant in $P_z$ is adjacent to $z$.
Moreover, $y$ is a center. Every vertex of the gadget of $y$ is at distance at most $3$ from $y$, every vertex of $S$ is adjacent to $y$, every blue vertex $z \in B$ is at distance at most $2$ from $y$, and every pendant in $P_z$ is at distance exactly $3$ from $y$ via a path $y-r-z-p$ for some $r \in S$ adjacent to $z$.
Hence $(G',d,\ell,u)$ is a yes-instance.

\medskip
\noindent
\textbf{Backward direction.}
Assume now that $(G',d,\ell,u)$ is a yes-instance.
Since $|V(G')|=2u$ and every cluster has size exactly $u$, there are exactly two clusters.

We first show that all vertices of the $x$-gadget lie in the same cluster as $x$.
Indeed, every pendant adjacent to $x$ has no neighbor other than $x$, so it must lie in the same cluster as $x$.
Now consider one of the two length-$3$ paths attached to $x$, say $x-x_1-x_2-x_3$.
If some vertex of this path were not in the same cluster as $x$, then the vertices of the path not containing $x$ would induce a connected component of size at most $3$ in that cluster.
Since every cluster has size exactly $u$ and $u>3$, such a component cannot form a cluster on its own, nor can it be connected to any other vertex outside the path.
Hence all vertices of the path must lie in the same cluster as $x$.
The same argument applies to the second path and also to the whole gadget of $y$.

Thus, if $x$ and $y$ were in the same cluster, then that cluster would contain
$x$, $y$, the whole gadget of $x$, and the whole gadget of $y$, and therefore would have size at least
$2+a+b = u+8 > u$,
a contradiction.
Hence $x$ and $y$ belong to different clusters. Let $C_x$ be the cluster containing $x$, and let $C_y$ be the cluster containing $y$.

We next show that $C_x$ contains no blue vertex.
Suppose that some $z \in B$ belongs to $C_x$.
Since every vertex of $P_z$ is adjacent only to $z$, all vertices of $P_z$ must also belong to $C_x$.
Therefore
$|C_x| \ge 1+a+1+c > u$,
a contradiction.
Thus $C_x$ contains no vertex of $B$ and no pendant adjacent to a blue vertex.
Consequently, $C_x$ consists only of $x$, its gadget, and some vertices of $R$.
Since $|C_x|=u$, we obtain from the equality
$1+a+(n-k)=u$
that $C_x$ contains exactly $n-k$ vertices of $R$.
Hence $C_y$ contains exactly $k$ vertices of $R$.

Now consider any blue vertex $z \in B$.
Since $z \notin C_x$, we have $z \in C_y$.
Moreover, every vertex of $P_z$ also belongs to $C_y$, because each such vertex is adjacent only to $z$.
We claim that $z$ has a neighbor in $C_y \cap R$.
Indeed, otherwise the vertices $z$ and $P_z$ would induce a connected component of $G'[C_y]$ disconnected from the rest of the cluster, since $z$ has neighbors only in $R$ and in $P_z$.
This contradicts the connectedness of $G'[C_y]$.
Therefore every blue vertex $z \in B$ has a neighbor in $C_y \cap R$.
Since $|C_y \cap R|=k$, the set $C_y \cap R$ is a red-blue dominating set of size exactly $k$ for the original instance $(G,k)$.
Therefore $(G,k)$ is a yes-instance.

This proves the equivalence. Since the parameter of the constructed instance is polynomially bounded in $|B|$, this is a polynomial parameter transformation from \textsc{Red-Blue Dominating Set} parameterized by $|B|$. The claimed kernel lower bound follows.
\end{proof}

\begin{theorem}
Unless coNP $\subseteq$ NP/poly, \textsc{Connected Network Microaggregation}
does not admit a polynomial kernel when parameterized by the vertex deletion
distance to a clique, even when $d=3$.
\end{theorem}

\begin{proof}
We give a polynomial parameter transformation from \textsc{Red-Blue Dominating Set}
parameterized by $|B|+k$.
Let $(G=(R \uplus B,E),k)$ be an instance, and set $n:=|R|$ and $t:=|B|$.
We may assume $k \le n$ and $t \ge 1$.
As preprocessing, we add isolated vertices to $R$ so that $n \ge 2k+t-9$.
This does not change the instance.
Define
$p := 9+n-2k-t \ge 0$.

We construct $(G',d,\ell,u)$ as follows.
First, turn $R$ into a clique. Add a set $P$ of $p$ new vertices and make
$R \cup P$ a clique.
We now add a gadget rooted at a vertex $x_0$.
Add vertices $x_1,\dots,x_8$ and edges:
\begin{itemize}
    \item $x_0x_1$, $x_1x_2$, $x_2x_3$,
    \item $x_0x_4$, $x_4x_5$, $x_5x_6$,
    \item $x_0x_7$, $x_7x_8$.
\end{itemize}
Finally, make $x_8$ adjacent to every vertex of $R$.
For every edge $rb \in E$, we keep the same edge in $G'$.
Set
$u := 9+n-k$, $\ell := u$, and $d:=3$.
Then
$k+t+p = 9+n-k = u$
and
$|V(G')| = n+t+p+9 = 2u$.

\medskip
\noindent
\textbf{Parameter bound.}
Let
$X := B \cup \{x_0,\dots,x_8\}$.
Then $|X| = t+9$, and $G'-X$ is the clique $R \cup P$.
Hence the parameter is bounded polynomially.

\medskip
\noindent
\textbf{Forward direction.}
Let $S \subseteq R$ be a dominating set of size $k$.
Define
$C_0 := \{x_0,\dots,x_8\} \cup (R \setminus S)$, and
$C_1 := S \cup B \cup P$.
Then $|C_0| = 9 + (n-k) = u$ and $|C_1| = k+t+p = u$.
The graph $G'[C_0]$ is connected and $x_0$ is a center:
every vertex is within distance at most $3$.
The graph $G'[C_1]$ is connected since $S \cup P$ is a clique
and $S$ dominates $B$.
Moreover, $x_8$ is a center: vertices of $S$ are adjacent,
vertices of $P$ are at distance $2$, and vertices of $B$
are at distance at most $2$ via $S$.

\medskip
\noindent
\textbf{Backward direction.}
Let $(G',d,\ell,u)$ be a yes-instance.
Let $C_0$ be the cluster containing $x_0$.
By connectivity arguments, all vertices $x_1,\dots,x_8$ must lie in $C_0$.
We now argue that $x_0$ must be the center of $C_0$.
Observe that $x_3$ and $x_6$ are at distance $6$, and the only vertex
within distance $3$ from both is $x_0$.
Hence $x_0$ is the unique possible center.
Therefore every vertex of $C_0$ must be at distance at most $3$ from $x_0$.
Vertices of $B$ and $P$ are at distance at least $4$ from $x_0$,
so they cannot lie in $C_0$.
Thus $C_0 \subseteq \{x_0,\dots,x_8\} \cup R$.
Since $|C_0|=u=9+n-k$, it follows that $C_0$ contains exactly $n-k$
vertices of $R$.

Hence the other cluster $C_1$ contains exactly $k$ vertices of $R$
together with all vertices of $B$ and $P$.
Since $G'[C_1]$ is connected, every $b \in B$ must have a neighbor in
$C_1 \cap R$, and thus this set forms a dominating set of size $k$.
This completes the reduction and proves the claim.
\end{proof}

\begin{theorem}\label{no kernel u+s}
Let $\mathcal{G}$ be a graph class such that for every $s \in \mathbb{N}$, the class $\mathcal{G}$ contains a connected graph on $s$ vertices with at least one pendant vertex. Unless coNP $\subseteq$ NP/poly, \textsc{Connected Network Microaggregation} does not admit a polynomial kernel when parameterized by the vertex deletion distance to $\mathcal{G}$ and $u$, even for all $d \ge u$.
\end{theorem}

\begin{proof}
We give a polynomial parameter transformation from \textsc{Red-Blue Dominating Set} parameterized by $|B|+k$.
Recall that an instance of \textsc{Red-Blue Dominating Set} consists of a bipartite graph $G=(R \uplus B,E)$ and an integer $k$, and the task is to decide whether there exists a set $S \subseteq R$ of size at most $k$ such that every vertex of $B$ has a neighbor in $S$.
It is known that this problem does not admit a polynomial kernel when parameterized by $|B|+k$, unless coNP $\subseteq$ NP/poly.

Let $(G=(R \uplus B,E),k)$ be an instance of \textsc{Red-Blue Dominating Set}.
Set $t:=|B|$ and $n:=|R|$. We may assume that $k \le n$ and $t \ge 1$, since otherwise the instance is trivial.
We construct an instance $(G',d,\ell,u)$ of \textsc{Connected Network Microaggregation} as follows.
Let $u := t+k+1$, $\ell := u$, and let $d \ge u$ be arbitrary.

For every vertex $r \in R$, choose a connected graph $H_r \in \mathcal{G}$ on exactly $u$ vertices with at least one pendant vertex. Let $q_r$ be a pendant vertex of $H_r$, and let $s_r$ be its unique neighbor in $H_r$.
For every edge $rb \in E$ of the original graph, we add the edge $q_r b$.
We add one new vertex $b^\star$ and make it adjacent to every vertex $q_r$, $r \in R$.
Finally, we add a set $Z=\{z_1,\dots,z_k\}$ of $k$ new vertices and make every $z_j$ adjacent to every vertex $s_r$, $r \in R$.
This completes the construction of $G'$.

\medskip
\noindent
\textbf{Parameter bound.}
Let $X := B \cup \{b^\star\} \cup Z$.
Then $|X| = t+1+k = u$.
Moreover, $G'-X$ is the disjoint union of the graphs $H_r$, $r \in R$.
Hence the vertex deletion distance of $G'$ to $\mathcal{G}$ is at most $u$.
Therefore the combined parameter, vertex deletion distance to $\mathcal{G}$ plus $u$, is at most $2u=2(t+k+1)$, which is polynomial in $|B|+k$.

\medskip
\noindent
\textbf{Forward direction.}
Assume that $(G,k)$ is a yes-instance of \textsc{Red-Blue Dominating Set}.
Then there exists a set $S \subseteq R$ of size at most $k$ dominating $B$.
By adding arbitrary vertices of $R\setminus S$ if necessary, we may assume that $|S|=k$.
Choose an arbitrary bijection $\varphi:S \to [k]$.
We define a partition of $V(G')$ into connected sets of size exactly $u$ as follows.

For every $r \in R\setminus S$, let $C_r := V(H_r)$.
For every $r \in S$, let $C_r := \bigl(V(H_r)\setminus\{q_r\}\bigr)\cup\{z_{\varphi(r)}\}$.
Finally, let $C^\star := B \cup \{b^\star\} \cup \{q_r : r\in S\}$.

We first check the sizes. For every $r\in R\setminus S$, we have $|C_r|=u$.
For every $r\in S$, we have $|C_r|=(u-1)+1=u$.
Also, $|C^\star| = |B|+1+|S| = t+1+k = u$.
Hence these sets form a partition of $V(G')$ into sets of size exactly $u$.

We now check connectedness.
Each set $C_r$ with $r\in R\setminus S$ induces the connected graph $H_r$, and hence is connected.
Each set $C_r$ with $r\in S$ is connected because $q_r$ is a pendant vertex of $H_r$, so $H_r-\{q_r\}$ is connected, and the vertex $z_{\varphi(r)}$ is adjacent to $s_r$.
Finally, $C^\star$ is connected: the vertex $b^\star$ is adjacent to every $q_r$ with $r\in S$, and every vertex $b\in B$ has a neighbor in $\{q_r:r\in S\}$ because $S$ dominates $B$ in the original instance.

Since every cluster has size exactly $u$ and is connected, its diameter is at most $u-1<d$.
Therefore any vertex of the cluster can serve as a center.
Hence $(G',d,\ell,u)$ is a yes-instance of \textsc{Connected Network Microaggregation}.

\medskip
\noindent
\textbf{Backward direction.}
Assume now that $(G',d,\ell,u)$ is a yes-instance. Since every cluster has size exactly $u$, the distance constraint is vacuous for all $d \ge u$, and we only need to reason about connectivity and size.
Since $\ell=u$ and every cluster has size exactly $u$, every feasible solution is a partition of $V(G')$ into connected sets of size exactly $u$.

\begin{claim}
For every $r\in R$, the set $V(H_r)\setminus\{q_r\}$ is contained in a single cluster.
\end{claim}

\begin{proof}
Suppose not. Since $H_r-\{q_r\}$ is connected, there exists a cluster $C$ that intersects $V(H_r)\setminus\{q_r\}$ but does not contain $s_r$.
Now every vertex of $C \cap (V(H_r)\setminus\{q_r\})$ has no neighbor outside $H_r$, and none of them is adjacent to $q_r$, because $q_r$ is a pendant vertex adjacent only to $s_r$.
Hence $C$ cannot contain any vertex outside $V(H_r)\setminus\{q_r\}$.
Therefore $|C| \le u-2 < u$, contradicting the fact that every cluster has size exactly $u$.
\end{proof}

\begin{claim}\label{claim 2}
For every $j\in[k]$, the vertex $z_j$ belongs to a cluster of the form $\bigl(V(H_r)\setminus\{q_r\}\bigr)\cup\{z_j\}$ for a unique $r\in R$.
\end{claim}

\begin{proof}
The vertex $z_j$ is adjacent only to the vertices $s_r$, $r\in R$.
Hence, for the cluster containing $z_j$ to be connected, it must contain some vertex $s_r$.
By Claim~1, this implies that it contains all of $V(H_r)\setminus\{q_r\}$.
Since $|V(H_r)\setminus\{q_r\}|=u-1$, the cluster can contain no further vertex.
Thus it is exactly $\bigl(V(H_r)\setminus\{q_r\}\bigr)\cup\{z_j\}$.
Uniqueness follows because two distinct such sets would together contain vertices from two distinct gadgets and therefore have size at least $2u-1>u$.
\end{proof}

\begin{claim}\label{claim 3}
Exactly $k$ gadgets are split, that is, there are exactly $k$ vertices $r\in R$ such that $q_r$ does not lie in the same cluster as $V(H_r)\setminus\{q_r\}$.
\end{claim}

\begin{proof}
By Claim~\ref{claim 2}, every vertex $z_j$ is paired with one set $V(H_r)\setminus\{q_r\}$, and different vertices $z_j$ are paired with different gadgets. Hence at least $k$ gadgets are split.

Conversely, if a gadget $H_r$ is split, then by Claim~1 the set $V(H_r)\setminus\{q_r\}$ must lie in one cluster.
Since this set has size $u-1$, that cluster must contain one additional vertex.
The only possible such vertex is some $z_j$, because vertices of $B\cup\{b^\star\}$ are adjacent only to $q_r$ and not to any vertex of $V(H_r)\setminus\{q_r\}$.
Hence every split gadget consumes one distinct vertex of $Z$, and therefore at most $k$ gadgets are split.
Thus exactly $k$ gadgets are split.
\end{proof}

By Claim~\ref{claim 3}, there are exactly $k$ gadgets whose pendant vertices are separated from the rest.
Let $S' := \{\, r\in R : q_r \text{ is not in the same cluster as } s_r \,\}$.
Then $|S'|=k$.
For every $r\notin S'$, the whole gadget $H_r$ must form a cluster of size $u$.
For every $r\in S'$, one cluster is $\bigl(V(H_r)\setminus\{q_r\}\bigr)\cup\{z_j\}$ for some $j\in[k]$, and the only vertex of $H_r$ not yet assigned is $q_r$.
Consequently, after removing all these clusters, the remaining vertices are exactly $B \cup \{b^\star\} \cup \{q_r : r\in S'\}$.
This set has size $|B|+1+|S'| = t+1+k = u$, and therefore it forms one final cluster, say $C^\star$.

We now show that $S'$ dominates $B$ in the original graph.
Consider any vertex $b\in B$.
Inside $C^\star$, the vertex $b$ is adjacent only to those vertices $q_r$ such that $rb\in E$ in the original graph.
It is not adjacent to any vertex of $B$, nor to $b^\star$.
Hence, if $b$ had no neighbor among the vertices $\{q_r:r\in S'\}$, then $b$ would be isolated in $G'[C^\star]$, contradicting the connectedness of $C^\star$.
Therefore every vertex of $B$ has a neighbor in $\{q_r:r\in S'\}$, and so $S'$ is a red-blue dominating set of size exactly $k$ for $(G,k)$.

Thus $(G,k)$ is a yes-instance if and only if $(G',d,\ell,u)$ is a yes-instance.
Since the parameter of the constructed instance is polynomially bounded in $|B|+k$, this is a polynomial parameter transformation.
Therefore, unless coNP $\subseteq$ NP/poly, \textsc{Connected Network Microaggregation} does not admit a polynomial kernel when parameterized by the vertex deletion distance to $\mathcal{G}$ and $u$.
\end{proof}

\begin{corollary}
Unless coNP $\subseteq$ NP/poly, \textsc{Connected Network Microaggregation}
does not admit a polynomial kernel when parameterized by the vertex integrity of the
input graph together with $u$, even for all $d \ge u$.
\end{corollary}

\begin{proof}
We use the same construction as in Theorem~\ref{no kernel u+s}.
Let $X := B \cup \{b^\star\} \cup Z$. By construction,
$|X| = |B| + 1 + k = u$.
Moreover, every connected component of $G' - X$ is one of the graphs $H_r$, and each
such graph has exactly $u$ vertices. Hence
\[
\mathrm{vi}(G') \le |X| + \max\{|V(C)| : C \text{ is a component of } G'-X\}
\le u + u = 2u.
\]
Therefore, $\mathrm{vi}(G') + u \le 3u = 3(|B|+k+1)$,
which is polynomial in the parameter of the input instance.
Since the reduction in Theorem~\ref{no kernel u+s} is correct, this yields a polynomial-
parameter transformation from \textsc{Red-Blue Dominating Set} parameterized by
$|B|+k$ to \textsc{Connected Network Microaggregation} parameterized by
$\mathrm{vi}+u$. The claimed kernel lower bound follows.
\end{proof}

\begin{corollary}\label{cor:paths-stars-kernel-lb}
Unless coNP $\subseteq$ NP/poly, \textsc{Connected Network Microaggregation}
does not admit a polynomial kernel under each of the following parameterizations:
\begin{enumerate}
    \item vertex deletion distance to paths together with $u$,
    \item vertex deletion distance to stars together with $u$,
\end{enumerate}
even for all $d \ge u$.
\end{corollary}

\begin{proof}
The statement follows directly from Theorem~\ref{no kernel u+s}.

Indeed, the class of paths contains a connected graph on $s$ vertices with a pendant vertex for every $s \in \mathbb{N}$, namely the path on $s$ vertices. Similarly, the class of stars contains such graphs, where every leaf is a pendant vertex.

Therefore, both classes satisfy the conditions of Theorem~\ref{no kernel u+s}, and the result follows.
\end{proof}

\section{W[1]-hardness Results}

\begin{theorem}\label{thm:gvd-hardness-large-d-gap}
Let $\mathcal{G}$ be a graph class such that for every $s \in \mathbb{N}$, the class $\mathcal{G}$ contains a connected graph on $s$ vertices. Then, for every fixed integer $r \ge 0$, \textsc{Connected Network Microaggregation} is $\mathrm{W}[1]$-hard when parameterized by the vertex deletion distance of the input graph to $\mathcal{G}$, even when $d \ge 2$ and $u-\ell=r$.
\end{theorem}

\begin{proof}
We provide a parameterized reduction from \textsc{Unary Bin Packing}. 
An instance of \textsc{Unary Bin Packing} consists of item sizes $a_1,a_2,\dots,a_m$, a number of bins $k$, and a bin capacity $B$, where all numbers are given in unary. The task is to decide whether the items can be partitioned into $k$ bins so that the total size of the items assigned to each bin is at most $B$. As usual, we may assume without loss of generality that every bin must have total size exactly $B$, by adding dummy items of size $1$ if necessary.

Fix $r \ge 0$, and let $c:=r+1$. We choose
$\alpha := (k+1)c + r + 1$.
Note that $\alpha > kc+r$ and $\alpha > r$.
Let $I$ be an instance of \textsc{Unary Bin Packing}. 
We construct an instance $I'=(G,d,\ell,u)$ of \textsc{Connected Network Microaggregation} as follows.
For every $j\in [k]$, we introduce a vertex $x_j$, and we add $r$ pendant vertices adjacent only to $x_j$. Hence each bin gadget contributes exactly $c=r+1$ vertices.
Next, for every item of size $a_i$, we choose a connected graph $H_i \in \mathcal{G}$ on $\alpha a_i$ vertices, and select an arbitrary vertex $q_i \in V(H_i)$. For every $i\in [m]$ and every $j\in [k]$, we add the edge $\{q_i,x_j\}$.

Now we introduce one special vertex $t$, together with $r$ pendant vertices adjacent only to $t$. Further, we choose a connected graph $Q \in \mathcal{G}$ on $\alpha B$ vertices. We make every vertex of $Q$ adjacent to $t$, and finally we make $t$ adjacent to every vertex of the graph constructed so far, except for its own pendant vertices and the pendant vertices attached to the vertices $x_j$. This completes the construction of the graph $G$.
We set $\ell := \alpha B + c$ and $u := \ell + r = \alpha B + c + r$. Let $d \ge 2$ be arbitrary.

The set $X:=\{x_1,\dots,x_k,t\}$ is a vertex deletion set of size $k+1$ such that $G-X$ is the disjoint union of the connected graphs $H_1,\dots,H_m,Q$, together with isolated vertices, all of which belong to $\mathcal{G}$ by assumption. Thus the parameter of the constructed instance is at most $k+1$, and the construction is clearly computable in polynomial time.

\medskip
\noindent
\textbf{Forward direction.}
Assume first that we are given a packing of the items into $k$ bins of capacity exactly $B$. For each $j\in [k]$, let $D_j$ denote the set of items assigned to bin $j$. We define one cluster for each bin by
$R_j:=\{x_j\}\cup \{\text{the $r$ pendants of }x_j\}\cup \bigcup_{i\in D_j} V(H_i)$,
and one additional cluster
$R_{k+1}:=\{t\}\cup \{\text{the $r$ pendants of }t\}\cup V(Q)$.

These sets are pairwise disjoint and cover all vertices of $G$. Each cluster is connected: for every item $i\in D_j$, the distinguished vertex $q_i$ is adjacent to $x_j$, and since $H_i$ is connected, every vertex of $H_i$ lies in the same connected component as $q_i$. Also, every pendant of $x_j$ is adjacent to $x_j$. Similarly, $R_{k+1}$ is connected since every vertex of $Q$ is adjacent to $t$, and every pendant of $t$ is adjacent to $t$.

It remains to verify the size and distance constraints. By construction,
$|R_j| = c + \sum_{i\in D_j} |V(H_i)| = c + \alpha\sum_{i\in D_j} a_i = c + \alpha B = \ell$,
and
$|R_{k+1}| = c + |V(Q)| = c + \alpha B = \ell$.
Hence every cluster has size in $[\ell,u]$.
Finally, the vertex $t$ can serve as a center for every cluster. Indeed, every vertex of every item gadget is adjacent either to $t$ or to some $x_j$ adjacent to $t$, and every pendant is at distance at most $2$ from $t$. Hence every vertex in every cluster lies at distance at most $2$ from $t$, and therefore the distance constraint is satisfied for every $d \ge 2$. Thus $I'$ is a yes-instance of \textsc{Connected Network Microaggregation}.

\medskip
\noindent
\textbf{Backward direction.}
Assume now that $(G,d,\ell,u)$ admits a feasible connected microaggregation.
We first observe that every cluster must contain at least one of the vertices $x_1,\dots,x_k,t$. Indeed, suppose a cluster $R$ contains none of these vertices. Then $R$ is contained entirely in some gadget $H_i$, in the gadget $Q$, or in a set of isolated pendant vertices. In all cases, $|R| \le \alpha B < \alpha B + c = \ell$, a contradiction.

For a cluster $R$, let $q(R)$ denote the number of vertices from $\{x_1,\dots,x_k,t\}$ contained in $R$, and let $S(R)$ denote the total original weight of the item gadgets fully contained in $R$.
We first claim that if $q_i \in R$, then $V(H_i) \subseteq R$. Suppose not, and let $v \in V(H_i)\setminus R$. Let $R'$ be the cluster containing $v$. Since $H_i$ is connected and $q_i$ is the only vertex of $H_i$ adjacent to vertices outside $H_i$, every path from $v$ to any vertex outside $H_i$ must pass through $q_i$. However, since $q_i \notin R'$, no such path exists entirely within $G[R']$, contradicting the fact that $R'$ induces a connected subgraph. Hence $V(H_i) \subseteq R$.
It follows that every cluster $R$ consists of exactly $q(R)$ anchor gadgets (each contributing $c$ vertices), together with a collection of entire item gadgets of total scaled size $\alpha S(R)$.

Therefore every cluster has size
$|R| = \alpha S(R) + q(R)c$.
Since $R$ is feasible, we have
$\ell \le |R| \le u$,
that is,
$\alpha B + c \le \alpha S(R) + q(R)c \le \alpha B + c + r$.
Equivalently,
$0 \le \alpha(S(R)-B) + (q(R)-1)c \le r$.
We claim that this implies $S(R)=B$ and $q(R)=1$. First, if $S(R)\ge B+1$, then
$\alpha(S(R)-B) + (q(R)-1)c \ge \alpha > r$,
a contradiction. On the other hand, if $S(R)\le B-1$, then
$\alpha(S(R)-B) + (q(R)-1)c \le -\alpha + kc < 0$,
because $q(R)\le k+1$ and $\alpha > kc+r$. Hence $S(R)=B$. But then the inequality reduces to
$0 \le (q(R)-1)c \le r$.
Since $c=r+1$, this is only possible when $q(R)=1$.
Thus every feasible cluster contains exactly one anchor gadget and item gadgets of total original weight exactly $B$.

In particular, the cluster containing $t$ contains no item gadget, because the graph $Q$ already contributes weight $B$. Hence it is exactly
$\{t\}\cup \{\text{the $r$ pendants of }t\}\cup V(Q)$.
The remaining $k$ anchor gadgets are the bin gadgets corresponding to $x_1,\dots,x_k$, and each of their clusters contains item gadgets of total original weight exactly $B$.
Consequently, the feasible clustering induces a partition of the items into $k$ groups of total size exactly $B$, which is a solution to the original \textsc{Unary Bin Packing} instance.

Therefore, the constructed \textsc{Connected Network Microaggregation} instance is a yes-instance if and only if the original \textsc{Unary Bin Packing} instance is a yes-instance. Since the parameter of the constructed instance is at most $k+1$, the reduction is parameter-preserving, and the claimed $\mathrm{W}[1]$-hardness follows.
\end{proof}

\begin{remark}
The above result is particularly strong as it shows $\mathrm{W}[1]$-hardness even under very permissive settings of the input parameters, namely for every fixed gap $u-\ell$. Thus, parameterizing by $u-\ell$ alone does not lead to fixed-parameter tractability. Moreover, the reduction works already for every $d \ge 2$.
\end{remark}

\begin{corollary}
\textsc{Connected Network Microaggregation} is NP-hard even on graphs of clique-width at most $4$, even when $d \ge 2$ and $u-\ell = r$ for any fixed $r \ge 0$.
\end{corollary}

\begin{proof}
We use the construction from Theorem~\ref{thm:gvd-hardness-large-d-gap}, with the additional choice that every graph $H_i$ and also the graph $Q$ are cliques. Since the reduction in Theorem~\ref{thm:gvd-hardness-large-d-gap} is correct for this special case as well, it suffices to show that every graph produced by this construction has clique-width at most $4$.
Recall that the constructed graph $G$ consists of the following parts:
\begin{itemize}
    \item $k$ bin vertices $x_1,\dots,x_k$, each with $r$ pendant vertices,
    \item for every item $a_i$, a clique $H_i$ on $\alpha a_i$ vertices with a distinguished vertex $q_i \in V(H_i)$,
    \item a clique $Q$ on $\alpha B$ vertices,
    \item one special vertex $t$ with $r$ pendant vertices,
\end{itemize}
where additionally each distinguished vertex $q_i$ is adjacent to every bin vertex, and $t$ is adjacent to every vertex except the pendant vertices.

We now describe a $4$-expression constructing $G$. Use labels $1,2,3,4$ with the following meaning: label $1$ for vertices that belong to already completed cliques, label $2$ for the bin vertices, and labels $3,4$ as temporary labels.
First, create the $k$ bin vertices and give them label $2$. For each $x_j$, create its $r$ pendant vertices one by one with label $4$, join labels $4$ and $2$, and relabel $4$ to $1$.
Next, for each item $i \in [m]$, construct the clique $H_i$ as follows. Create the distinguished vertex $q_i$ with label $4$, join labels $4$ and $2$, and relabel $4$ to $3$. Then create the remaining $\alpha a_i-1$ vertices of $H_i$ one by one: whenever a new vertex is created with label $4$, join labels $4$ and $3$, and then relabel $4$ to $3$. Finally, relabel all vertices of $H_i$ from label $3$ to label $1$.

Construct the clique $Q$ in the same way: create one vertex with label $3$, then repeatedly create a new vertex with label $4$, join labels $4$ and $3$, and relabel $4$ to $3$, until $Q$ is complete. Afterwards relabel all vertices of $Q$ from label $3$ to label $1$.
Next, create the vertex $t$ with label $4$ and join label $4$ to labels $1$ and $2$. Then create the $r$ pendant vertices of $t$ one by one with label $4$, join them to $t$, and relabel them to $1$.
It is immediate from the construction that the obtained graph is exactly the graph produced in Theorem~\ref{thm:gvd-hardness-large-d-gap}. Since we used only four labels, the clique-width of $G$ is at most $4$.

Hence the reduction of Theorem~\ref{thm:gvd-hardness-large-d-gap} yields NP-hardness even on graphs of clique-width at most $4$, and the hardness holds even when $d \ge 2$ and $u-\ell = r$.
\end{proof}

\begin{corollary}\label{cor:instantiations-gvd}
For every fixed integer $r \ge 0$, \textsc{Connected Network Microaggregation} is $\mathrm{W}[1]$-hard under each of the following parameterizations, even when $d \ge 2$ and $u-\ell=r$:
\begin{enumerate}
    \item vertex deletion distance to paths,
    \item vertex deletion distance to stars,
    \item cluster vertex deletion number,
    \item feedback vertex set number,
    \item treedepth.
\end{enumerate}
In particular, the problem is also $\mathrm{W}[1]$-hard when parameterized by pathwidth or treewidth.
\end{corollary}
\begin{proof}
We apply Theorem~\ref{thm:gvd-hardness-large-d-gap} with suitable choices of the graph class $\mathcal G$.

If $\mathcal G$ is the class of paths, we obtain $\mathrm{W}[1]$-hardness parameterized by vertex deletion distance to paths.

If $\mathcal G$ is the class of stars, we obtain $\mathrm{W}[1]$-hardness parameterized by vertex deletion distance to stars. Moreover, deleting the vertices $x_1,\dots,x_k,t$ leaves a forest, and hence the feedback vertex set number of the constructed graph is at most $k+1$. This yields hardness for feedback vertex set.

Furthermore, in the star instantiation, the resulting graph has treedepth bounded by a function of $k$. Hence the problem is $\mathrm{W}[1]$-hard parameterized by treedepth.

If $\mathcal G$ is the class of cliques, we obtain $\mathrm{W}[1]$-hardness parameterized by cluster vertex deletion number, since deleting $x_1,\dots,x_k,t$ leaves a disjoint union of cliques.

Finally, hardness for treedepth implies hardness for pathwidth and treewidth.
\end{proof}

\begin{theorem}\label{thm:ecp-to-cnma}
There is a polynomial-time parameter-preserving reduction from \textsc{Equitable Connected Partition} to \textsc{Connected Network Microaggregation}. In particular, given an instance $(G,r)$ of \textsc{Equitable Connected Partition} with $|V(G)|=n$ and $r \mid n$, one can construct in polynomial time an equivalent instance $(G,d,\ell,u)$ of \textsc{Connected Network Microaggregation} by setting
$d:=n$, $\ell:=n/r$ and .
\end{theorem}

\begin{proof}
Let $(G,r)$ be an instance of \textsc{Equitable Connected Partition}, where $|V(G)|=n$ and $r \mid n$. We construct an instance $(G,d,\ell,u)$ of \textsc{Connected Network Microaggregation} on the same graph $G$ by setting $d:=n$ and $\ell=u:=n/r$. 
We show that the two instances are equivalent.

Assume first that $(G,r)$ is a yes-instance of \textsc{Equitable Connected Partition}. Since $r \mid n$, there is a partition $V(G)=V_1 \uplus \cdots \uplus V_r$ such that $|V_i|=n/r$ for every $i \in [r]$ and $G[V_i]$ is connected for every $i \in [r]$. We claim that $(V_1,\dots,V_r)$ is a feasible solution for the constructed \textsc{Connected Network Microaggregation} instance. Indeed, each part has size exactly $\ell=u=n/r$, and each induced subgraph $G[V_i]$ is connected. It remains to verify the distance condition. Since $G[V_i]$ is connected and $|V_i|\le n$, the distance between any two vertices in $V_i$ is at most $n-1$. Hence any vertex of $V_i$ may serve as a center, because every vertex of $V_i$ is at distance at most $n-1 < n = d$ from it. Therefore $(G,d,\ell,u)$ is a yes-instance.

Conversely, assume that $(G,d,\ell,u)$ is a yes-instance of \textsc{Connected Network Microaggregation}. Then there exists a partition $V(G)=C_1 \uplus \cdots \uplus C_t$ such that each $C_j$ induces a connected subgraph of $G$ and satisfies $\ell \le |C_j| \le u$. Since $\ell=u=n/r$, it follows that every cluster has size exactly $n/r$. As the clusters partition $V(G)$, we obtain
$t \cdot (n/r) = n$,
and hence $t=r$. Therefore $C_1,\dots,C_r$ form a partition of $V(G)$ into $r$ connected parts of equal size, which is a feasible solution to the \textsc{Equitable Connected Partition} instance $(G,r)$.

Thus, $(G,r)$ is a yes-instance if and only if $(G,d,\ell,u)$ is a yes-instance. Since the reduction keeps the graph unchanged, any graph parameter of the input graph is preserved.
\end{proof}

\begin{corollary}\label{cor:fes-hard}
\textsc{Connected Network Microaggregation} is $\mathrm{W}[1]$-hard when parameterized by the feedback edge set number of the input graph.
\end{corollary}

\begin{proof}
It is known that \textsc{Equitable Connected Partition} is $\mathrm{W}[1]$-hard when parameterized by the feedback edge set number of the input graph~\cite{10.1007/978-3-642-11269-0_10}, even on instances satisfying $r \mid |V(G)|$. By Theorem~\ref{thm:ecp-to-cnma}, such an instance $(G,r)$ can be transformed in polynomial time into an equivalent instance $(G,d,\ell,u)$ of \textsc{Connected Network Microaggregation} on the same graph $G$. Since the graph is unchanged, the feedback edge set number is preserved. Hence \textsc{Connected Network Microaggregation} is $\mathrm{W}[1]$-hard under this parameterization.
\end{proof}

\section{Conclusion}
We conducted a systematic study of the parameterized complexity of \textsc{Connected Network Microaggregation} in the unweighted setting, focusing on the interplay between structural parameters and natural clustering parameters such as $d$, $u$, and $\ell$. Our results reveal a clear dichotomy. On the positive side, combining structural parameters with the cluster size bound $u$ yields fixed-parameter tractability, and even polynomial kernels in the case of vertex cover. On the negative side, structural parameters alone are insufficient: the problem remains $\mathrm{W[1]}$-hard and does not admit polynomial kernels under several natural parameterizations. Notably, these hardness results are robust and already hold under the strong restriction $d = 1$ and $\ell = u$, highlighting that neither tight distance bounds nor exact cluster sizes simplify the problem. 

These results naturally raise several directions for future work. While we obtain a polynomial kernel for $\mathrm{vc} + u$, it remains open whether we can obtain a polynomial kernel when parameterized by structural parameter such as twin-cover or neighbourhood diversity combined with $u$.  Finally, given the strong lower bounds for exact computation, exploring approximation or parameterized approximation algorithms for \textsc{CNMA} appears to be a promising direction, as does resolving the remaining open cases in the complexity landscape.

\bibliography{References}

\newpage

\end{document}